%% file: main_v5.tex
\documentclass[11pt]{article}
\usepackage{color}
\usepackage{epsfig}
\usepackage{epstopdf}
\usepackage[section]{placeins}
\usepackage{graphicx}
\usepackage{amsmath}
\usepackage{amssymb}
\usepackage{amsfonts}
\usepackage{url}
\usepackage{multirow}
\usepackage{array}
\usepackage{latexsym}
\usepackage{natbib}
\usepackage[noline,ruled]{algorithm2e}
\usepackage{bigstrut}
\usepackage{subcaption}
\usepackage{mathrsfs}
\usepackage{scalefnt}

\usepackage[toc,page]{appendix}
\usepackage{comment}
\usepackage[]{mdframed}

\newsavebox{\theorembox}
\newsavebox{\lemmabox}
\newsavebox{\claimbox}
\newsavebox{\factbox}
\newsavebox{\corollarybox}
\newsavebox{\examplebox}
\newsavebox{\remarkbox}
\newsavebox{\assbox}
\newsavebox{\propositionbox}
\newsavebox{\problembox}
\newsavebox{\defbox}

\savebox{\theorembox}{\noindent\bf Theorem}
\savebox{\lemmabox}{\noindent\bf Lemma}
\savebox{\factbox}{\noindent\bf Fact}
\savebox{\corollarybox}{\noindent\bf Corollary}
\savebox{\examplebox}{\noindent\bf Example}
\savebox{\remarkbox}{\noindent\bf Remark}
\savebox{\assbox}{\noindent\bf Assumption}
\savebox{\propositionbox}{\noindent\bf Proposition}
\savebox{\problembox}{\noindent\bf Problem}
\savebox{\defbox}{\noindent\bf Definition}

\newtheorem{ass}{\usebox{\assbox}}
\newtheorem{thm}{\usebox{\theorembox}}

\newtheorem{lem}{\usebox{\lemmabox}}

\newtheorem{cor}{\usebox{\corollarybox}}

\newtheorem{defn}{\usebox{\defbox}}

\def\blackslug{\hbox{\hskip 1pt \vrule width 4pt height 8pt depth 1.5pt
\hskip 1pt}}

\newcommand{\qed}{\mbox{}\hspace*{\fill}\nolinebreak\mbox{$\rule{0.7em}{0.7em}$}}

\newenvironment{proof}[1][\noindent \bf Proof:]{\begin{trivlist}
\item[\hskip \labelsep {\bfseries #1}]}{\end{trivlist}}

\input{symbols.tex}

\allowdisplaybreaks
\begin{document}

%
%
%
%
%

\title{Mind the Gap in the Mining Game}

\author{Kyoung-Kuk Kim\thanks{College of Business, E-mail: {\tt kkim@kaist.ac.kr}}, \quad Donghwa Seo\thanks{Corresponding author, Industrial and Systems Engineering, E-mail: {\tt properitas95@kaist.ac.kr} } \\
\small{\it Korea Advanced Institute of Science and Technology}
\\
}

\date{Sep 2024 }

\maketitle

\baselineskip 18pt
\begin{abstract}
We analyze intentional block delays (mining gaps) in Proof-of-Work blockchain systems, where miners strategically balance mining rewards against operational costs. Using a game-theoretic model, we derive a Nash equilibrium with optimal mining strategies and establish necessary and sufficient conditions for mining gap existence. We demonstrate that mining gaps, when combined with difficulty adjustment algorithms, can destabilize the system. We propose conditions to address sustainability concerns as block rewards decrease and reliance on transaction fees increases. Our findings are illustrated through a two-player game simulation and an analysis of the Bitcoin network, providing insights for blockchain design and policy. This work contributes to understanding strategic mining behavior and its impact on blockchain stability and efficiency.

\vs
\end{abstract}

\section{Introduction}
Cryptocurrencies and blockchain technology have been one of major innovations in the recent history of finance. The core mechanism is decentralized ledger technology, where network nodes maintain a consistent ledger through consensus algorithms. In blockchain networks, data blocks are linked in a chain, each referencing its predecessor through encrypted connections. Block creation is based on consensus protocols, which impact system scalability, sustainability, and stability. Proof-of-Work (PoW) is the most prominent consensus algorithm, supporting cryptocurrencies like Bitcoin, Dogecoin, and Bitcoin Cash.

The economies and design issues of these systems are not fully understood, as noted by \cite{biais2023}. Network participants' aggregate decisions determine system dynamics, and we need to understand the trade-offs among their strategic options. Block creation involves solving complex puzzles defined by the PoW protocol for economic rewards. Various strategic actions or attacks have been studied in the literature. Double spending, executable by entities controlling 51\% of network computing power, is a well-known attack \citep{kroll2013economics, gervais2016security}. Selfish mining, introduced by \cite{eyal2018majority}, involves colluding miners secretly minting blocks for higher profits. Subsequent works by \cite{nayak2016stubborn}, \cite{sapirshtein2017optimal}, and \cite{hansjoerg2022profitability} discuss optimal strategies and attack combinations. Even without hiding the block, miners' strategic actions can make the system vulnerable. Mining is incentivized by block rewards and transaction fees. It is economically viable only when revenue exceeds costs, including hardware investments and operational expenses like electricity. A critical question is whether systems like Bitcoin can remain sustainable with transaction fees alone, given periodic halving of block rewards.

\cite{carlsten2016instability} demonstrate that miners may cease operations unless accumulated fees are sufficient to cover costs. This non-performing computing power, so called mining gap, increases system vulnerability. Building on this, \cite{tsabary2018gap} provide valuable insights through numerical simulations, predicting gap occurrence when block rewards are sufficiently low compared to fees.
Our work extends this research by developing a rigorous analytical framework. We consider multiple miners competing for rewards, each controlling their computing power utilization. Through mathematical analysis, we derive explicit conditions for the inevitability of mining gaps and establish sustainability criteria as block rewards decrease and reliance on transaction fees increases. 

Our key contribution is a theoretical examination of how mining gaps interact with difficulty adjustment algorithms (DAA). The DAA regulates block generation speed towards a pre-defined value, adjusting puzzle difficulty based on network hash power. However, we show that the interaction between strategic mining behavior and DAA can lead to unexpected outcomes. This builds on work by \cite{fiat2019energy}, who demonstrate periodic miner slowdowns to reduce complexity, and studies on `smart mining' and `fickle mining' by \cite{goren2019mind} and \cite{kwon2019bitcoin}. Such an interaction may lead to system instability also through economically motivated and simpler actions of hash power operations. These results complement and extend the numerical findings of previous studies, providing a deeper understanding of blockchain system dynamics. 

Lastly, we complete our analysis with numerical experiments. Firstly, we use a two-player setting where there are a large miner and a small miner in terms of capacity. And we examine what happens to the equilibrium dynamics of the system when various network features such as the dominance of the large miner change. As a by-product, it is shown that hash power concentration incurs more mining gaps but larger utilities for both miners. Indeed, there are economic incentives for miners to coalesce as shown in Section~\ref{subsec:coalition}. Secondly, we apply our results to the Bitcoin network. The halving mechanism together with DAA may result in mining gaps which would make the system vulnerable. We present the stability criteria of block rewards and transaction fees. 


Our model assumes an increasing, concave function of expected rewards over time, reflecting collective user actions. While we don't fully explore user behavior, we discuss implications of our findings for user strategies regarding transaction fees. This is connected to the works on transaction fee dynamics and order arrivals by \cite{brenner2015trends}, \cite{bowden2020modeling}, \cite{kawase2017transaction}, \cite{kasahara2018effect}, \cite{easley2019mining}, and \cite{gebraselase2021transaction}, as well as priority queueing games studied by \cite{huberman2021monopoly} and \cite{li2022analyzing}.

The paper is structured as follows: Section~\ref{sec:Model} presents our model for block generation delays. Section~\ref{sec:Gap Game} analyzes miners' optimal strategies and Nash equilibrium existence. Section~\ref{sec:DAA} examines the existence of stable limits in the long run. Sections~\ref{sec:Simulation} and \ref{sec:Case Study} illustrate findings through a two-player game and Bitcoin network analysis. Section~\ref{sec:Conclusion} discusses implications for blockchain design and future research directions. The proofs of lemmas and theorems are deferred to the appendix. 

\section{The Model}\label{sec:Model}
\subsection{Problem Description}
We consider a continuous-time setting with $N$ miners. Each miner controls his/her mining power from zero to the maximum, say $\gamma_n$ for miner $n$, to maximize expected profits. The block generation time is random with a finite mean, reflecting the probabilistic nature of the PoW protocol. For simplicity, we disregard propagation delays of blocks and transactions, assuming all miners observe the same queue of pending transactions. This allows us to focus on miners' strategic decisions.
We also assume that the queue of pending transactions is continuously sorted in the descending order of transaction fees. At most the top $K$ transactions can fit in a block, where $K$ is determined by the block size limit. When a block is created, the top $K$ transaction requests are removed from the queue. This assumption reflects the typical behavior of miners who prioritize transactions with higher fees to maximize the revenue.

The model incorporates several exogenous factors influencing miner behavior. The function $R(t)$ represents the expected revenue from successful block generation when time $t$ has elapsed since the last block generation. This includes both the block reward and the fees from the transactions in the block. We denote the cost per hash operation by $c$, which is also in the currency of the target network. The probability of successful mining per hash operation, denoted as $w$, is inversely proportional to the PoW difficulty. This winning rate determines the likelihood of block generation.

    \begin{figure}[t]
        \centerline{\includegraphics[width=0.8\textwidth]{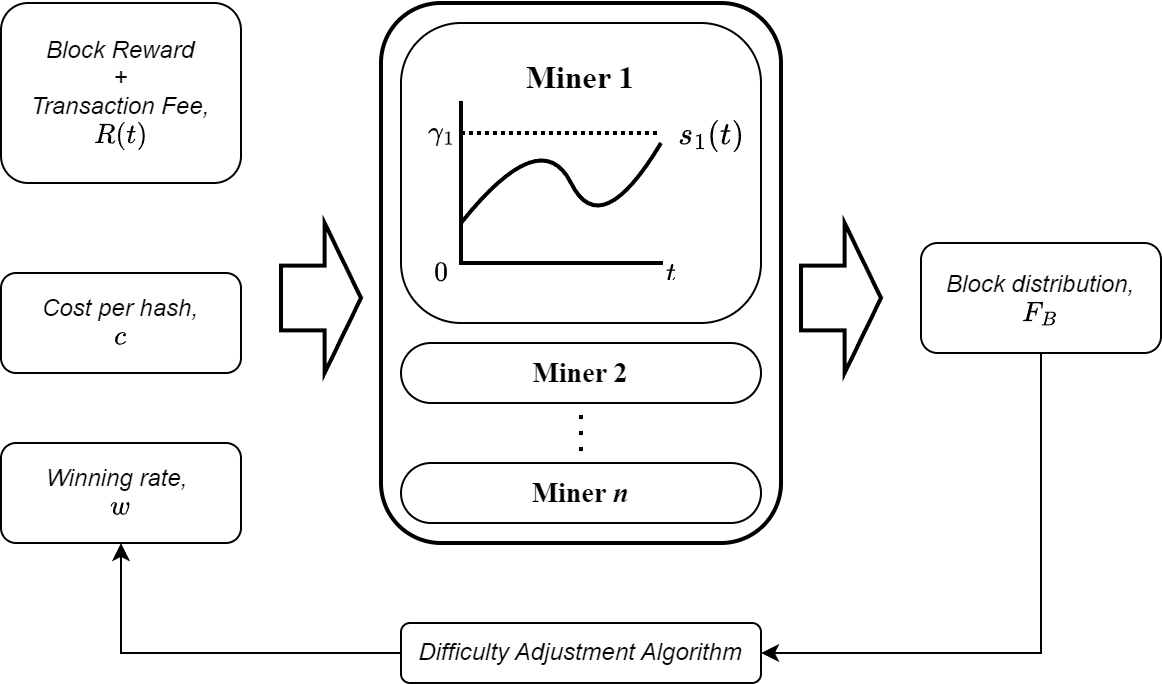}}
        \caption[Model]{The feedback process between model parameters, block distribution, and DAA.} \label{fig:2.2 model}
    \end{figure}

Miner $n$ establishes a hash supply function $s_n(t)$, representing his mining rig utilization at time $t$, bounded by $\gamma_n$. The strategic decisions of all miners collectively determine the probabilistic characteristics of new block generation. The system's DAA then adjusts $w$ to maintain a constant expected block generation time $\mu$ (e.g., 10 minutes for Bitcoin). Figure~\ref{fig:2.2 model} illustrates this feedback process in a single diagram. We define the total mining power function $s(t) = \sum_{n=1}^N s_n(t)$, total capacity $\gamma = \sum_{n=1}^N \gamma_n$, and an auxiliary function $s_{-n}(t) = s(t) - s_n(t)$ for each miner $n$.

\begin{ass}\label{assumption:liminf}
For each miner $n$, the long run limit of $s_{-n}(\cdot)$ is positive, or $\liminf_t s_{-n}(t)>0$. 
\end{ass}

This assumption reflects the belief that at least some miners will eventually be operating due to accumulating transaction fees. It also implies that a new block is generated in finite time. This belief is reasonable because if mining ceases for an extended period, the accumulation of fees would incentivize individual miners to resume operations. 
The potential for high rewards would likely attract opportunistic miners, ensuring some level of continuous mining activity in the long run.
\subsection{Block Distribution and Revenue Function}
The block generation time $B$ follows a distribution that is locally a geometric distribution because of the nature of hash operations in mining. The likelihood of block generation or hazard rate function at $t$ can be written as $f_B(t)/\bar F_B(t) = w \cdot s(t)$ where $f_B$ is the probability density function (PDF) and $\bar F_B$ is the complementary cumulative distribution function (CCDF) of the block generation time $B$. This likelihood is proportional to the winning rate $w$ and the total hash power $s(t)$ at the time. Then, clearly, 
\begin{equation}\label{eq:F_B}
\Bar{F}_B(t)  =     \bbP \left( {\rm Exp}(1) > w\int_0^t s(u)\rd u\right) = \exp\left(-w\int_0^t s(u)\rd u\right).
\end{equation}
Here ${\rm Exp}(1)$ is an exponential distribution with mean 1. When a new block is generated at time $t$ (the elapsed time since the last block generation), the successful miner receives the block reward $BR$ and the transaction fees. As noted earlier, those fees are the top $K$ fees attached to transactions and we write $b_{(i)}(t)$ for the $i$-th largest transaction fee in the queue at time $t$. This means the expected revenue from successful mining at time $t$ is given by $R(t) = BR +\bbE \left[\sum_{i=1}^{K}b_{(i)}(t)\right]$.



We model transaction arrivals as a Poisson process with i.i.d. fees. The incremental expected revenue is determined by new transaction arrivals with fees greater than the current $K$-th largest fee:
\begin{eqnarray*}
R(t+h)-R(t) &=& \bbE\left[ \sum_{i=1}^K b_{(i)}(t+h) - \sum_{i=1}^K b_{(i)}(t) \right] = \lambda h \bbE\left[\left(b-b_{(K)}(t)\right)^+ \right] + o(h). 
\end{eqnarray*}
Here $b$ is the random transaction fee and $\lambda h$ is the probability of a transaction arrival for a small time $h$. The incremental expected revenue is $\left(b-b_{(K)}(t)\right)^+ $ because for an arriving transaction to increase $R$, it must have a fee larger than $b_{(K)}(t)$. From the assumption of i.i.d. arrivals and the fixed block size $K$, we can derive the following property about the $R(t)$ which is imposed as an assumption. 

\begin{ass}\label{assumption:revenue}
The expected revenue $R(t)$ conditional on $B = t$ is twice differentiable, strictly increasing, and concave. In addition, this conditional expected revenue per hash operation $w R(t)$ exceeds per hash operating cost $c$ in the limit. 
\end{ass}
Note that $R(t)$ is increasing because $\bbE\left[\left(b-b_{(K)}(t)\right)^+ \right]>0$. It is concave because $b_{(i)}(t)$ is non-decreasing in $t$ for any transaction sequence and any $i$, leading to non-increasing incremental expected revenue over time.
\subsection{The Mining Game}
The unconditional expected operating profit of miner $n$ is given by 
\begin{eqnarray*}
J_n &=& \int_0^\infty \bbE[ J_n | B = t] f_B(t) \rd t \\
&=& \int_0^\infty \left\{ \frac{s_n(t)}{s(t)} R(t)  - c \int_0^t s_n(u) {\rm d}u \right\} f_B(t) \rd t  \\
&=& \int_0^\infty (w \; R(t) - c) s_n(t) \bar F_B(t) \rd t. 
\end{eqnarray*}
Since $R(t)$ is an increasing function, $\lim_{t \rightarrow \infty} (w R(t) -c )$ exists. If this value is not positive, then it is optimal for miners to shut down, i.e. $s_n(t) = 0$ for any $t$ and any $n$. Thanks to Assumption~\ref{assumption:revenue}, this limit is strictly positive, i.e. $\lim_t R(t) > c/ w$ and we avoid triviality. Next, we rewrite the problem of optimizing $J_n$ as an optimal control problem by introducing a new state variable $x(t) = \int_0^t s(u) \rd u = \int_0^t s_n(u) \rd u+\int_0^t s_{-n}(u) \rd u$ and the control variable $s_n(t)$. Here $x(t)$ represents the number of hash operations attempted up to time $t$. Using \eqref{eq:F_B}, the problem is expressed as
\begin{equation}\label{eq:target}
\begin{split}
\max_{s_n} \quad & J_n = \int_0^\infty \left(w R(t)-c\right)s_n(t)e^{-w x(t)}\rd t  \\
{\rm s.t.} \quad & x'(t) = s(t), \\
& x(0)  = 0, \\
& s_n(t)  \in [0,\gamma_n ].  
\end{split}
\end{equation}

We begin with no particular constraint on the form of $s_n(\cdot)$ except measurability, although we focus on more specific and realistic functional forms in later sections. The state variable is assumed to be absolutely continuous. Note that $s(t)=\sum_{i=1}^N s_n(t)$, thus the behavior of each miner affects the distribution of block generation, making \eqref{eq:target} a differential game. The only observable information for a miner is whether a block is mined or not. However, miners can infer changes in the behavior of other miners by observing changes in the block distribution ex post. Because of this limited real-time information, we consider an open-loop Nash equilibrium. 
Lastly, our model focuses on the dynamics of dominant mining pools rather than individual miners who have little influence on the block distribution. For this reason, we fix the number of miners $N$. This is also because the emergence of a significantly large mining pool has been rare, e.g. just a handful of influential mining pools in the history of the Bitcoin network.  

\begin{lem}\label{lem:existence}
There exists an optimal control for \eqref{eq:target} for any given $s_{-n}(t)$.
\end{lem}

Recall that Assumptions~\ref{assumption:liminf} and \ref{assumption:revenue} apply throughout the paper. The result of \cite{dmitruk2005existence} is invoked for a proof of this lemma which can be found in the appendix. In the section that follows, we restrict our attention to a certain type of mining power functions which are plausible and implementable in practice. Equilibrium strategies of miners in such a strategy space help us better understand the complex dynamics of blockchain systems.  

\section{Strategic Delay in Block Generation} \label{sec:Gap Game}
In this section, we study the properties of optimal solutions to \eqref{eq:target} and extend our analysis to a game setting where miners choose their own controls $s_n$ for exogenously given model features such as the conditional revenue function $R(t)$, cost rate $c$, and winning rate $w$. The operating cost and winning probability depend on these rates multiplied by the input computing resources. This game involves strategic interactions where miners aim to maximize their profits $J_n$ in \eqref{eq:target}, taking into account the actions of other miners. The strategic nature of this game arises from the fact that a miner's success in mining a block depends not only on their own mining power but also on the total mining power in the network.

The first part of this section characterizes the optimality conditions based on classical deterministic control theory and provides insights into the strategies miners could use to maximize their profits. The second part examines the existence and characteristics of Nash equilibrium in the model, focusing on how the interplay between individual decisions and collective outcomes shapes the dynamics of the mining process. Lastly, all the strategic interactions are clearly demonstrated in the setting of homogeneous miners with affine $R(t)$. 

\subsection{Optimal Control of Miner}
Let us denote the optimal control of miner $n$ by $s_n^*$ and the corresponding state by $x^*$ to distinguish it from the uncontrolled state $x$. The Hamiltonian $H_n$ of \eqref{eq:target}, the adjoint variable $\psi$, the switching function $\sigma_n$, and the auxiliary function $\xi_n$ are defined as follows: 
\begin{equation*}
\begin{aligned}
H_n(x,s_n, \psi, t) &= s_n(t) (w R(t) - c)e^{-w x(t)} + \psi(t) s(t), \\
\psi'(t) &= w s_n^*(t) (w R(t) - c) e^{-w x^*(t)}, \\
\sigma_n(t) &:= (w R(t) - c) e^{-w x^*(t)} + \psi(t), \\
\xi_n(t) &:= R'(t) - (w R(t) - c) s_{-n}(t).
\end{aligned}
\end{equation*}
The maximum principle states that $H_n$ is maximized at $x^*$, $s_n^*$ and it holds that 
$H_n(x^*, s_n^*, \psi, t) \geq H_n(x, s_n, \psi, t)$ for any $x(t)$ and $s_n(t)$. By the maximum principle, $\psi'(t) = - \partial H_n / \partial x$ must be satisfied for an optimal $H_n$. This condition leads to the second equation for $\psi'(t)$. The transversality condition $\lim_{t \rightarrow \infty} \psi(t) = 0$ allows us to express $\psi(t)$ as an integral:
\begin{equation*}
\psi(t) = - \int_t^\infty w s_n^*(u) (wR(u) - c) e^{-w x^*(u)} \rd u.
\end{equation*}

Now we observe that $H_n$ is linear in $s_n$, so that we re-write it as $H_n = s_n(t) \sigma_n(t) + s_{-n}(t) \psi(t)$. This linearity implies that the $s_n^*$ maximizing $H_n$ is determined by the sign of $\sigma_n$. Consequently, the optimal control $s_n^*(t)$ is of bang-bang type; i.e., 
$s_n^*(t) = \gamma_n$ if $\sigma_n(t) > 0$ and 0 if $\sigma_n(t) < 0$. 
It is important to note that the maximum principle does not apply when $\sigma_n(t) = 0$. This situation, known as singular control, requires a separate analysis. The bang-bang nature of the control outside of singular intervals suggests that miners will either operate at full capacity or not at all, depending on the sign of the switching function. The following lemma gives us a hint about the optimal behavior of the miner for intervals where the sign of $\sigma_n$ is known.

\begin{lem}
    \label{lem:long_term_s}
        For each miner, an optimal control $s_n^*$ is zero almost everywhere on  $\{t | w R(t)<c\}$, but we have $s_n^*(t)=\gamma_n$ for all sufficiently large $t$ values. 
\end{lem}   

As a profit-seeking agent, miners turn off the mining rigs when the expected revenues are clearly smaller than the costs. This may happen at the beginning of a new block generation. However, they eventually turn on the machines to their maximum capacities as it becomes profitable due to the accumulated fees in the queue. To find optimal control in-between, one must analyze the sign of the switching function and consider the singular control case. 

\begin{lem}
\label{lem:singular_interval}
Suppose there exists an interval on which $\sigma_n(t) = 0$. Then the optimal utility of miner $n$ is not affected by any change of $s_n^*$ on that interval.
\end{lem}

This result implies that the bang-bang control is among optimal controls when singular control occurs on intervals, i.e. optimal $s_n^*(t)$ takes on its minimum or maximum values $0$ and $\gamma_n$. As in the next lemma, one can further show the optimality of specific bang-bang controls for any time $t$. Therefore, miner $n$ only needs to consider the optimal timing of turning on his entire mining rig. 

\begin{lem}
\label{lem:increasing_s}
Suppose that $s_{-n}$ is non-decreasing. Then, there exists an optimal control $s_n^*$ that is $\gamma_n {\bf 1}_{{t \geq t^*}}$ for some $t^*$.
\end{lem}

Partially operating or decreasing mining capacity is therefore sub-optimal. If such a strategy were optimal, it would mean less valuable future rewards than the current expected revenue. However, accumulated future rewards are actually greater given unsuccessful mining. For the same reason, the underlying assumption of non-decreasing $s_{-n}$ is plausible; e.g., if it is optimal to turn on machines at time $t$, then miners would opt to provide non-decreasing computing power at $t+\varepsilon$ given unsuccessful mining up to time $t$. The result in this subsection shows the optimal response of miner $n$, given the hash supply of other miners. The next subsection deals with optimal responses of multiple miners and studies the existence and characterization of their equilibrium behaviors.

\subsection{Open-loop Nash Equilibrium}\label{subsec:nash}

In a setting where $N$ miners compete for block rewards and accumulated fees, each miner's optimal control depends on the total hash power $s(t)$, which is the collective outcome of miners' strategic decisions. Each hash supply $s_n(\cdot)$ is determined by the physical manipulation of mining machines. Therefore, it is natural to consider this function as a piecewise constant function that cannot jump infinitely often within a finite interval.  This assumption aligns with the realistic limitations of mining hardware and helps streamline our analysis. The following lemma addresses simple mining strategies, where miners activate their machines to full capacity after a specific time.

\begin{lem}\label{lem:simple_s} 
Let us assume $\gamma_1 \leq \gamma_2 \leq\cdots \leq \gamma_N$. Further suppose that controls specified by $s_n(t) = \gamma_n {\bf 1}_{\{t \geq t_n^*\}}$ for each $n$ are optimal. Then, we have $t_1^* \leq t_2^* \leq \cdots \leq t_N^*$ and they are uniquely determined by $\sigma_n(t_n^*) = 0$ for each $n$. If there is no solution to $\sigma_n(t) = 0$, then $t_n^* = 0$. 
\end{lem}
This lemma demonstrates that miners with higher capacities tend to start mining later, aligning with the intuition that larger miners can afford to wait for potentially higher rewards. It also provides a formal proof of the simulation-based findings of \cite{tsabary2018gap}.

Our main result on the strategic decisions of $N$ miners is presented in the following theorem. A Nash equilibrium is written as a collection of strategies $s_n^*$'s such that each miner $n$ does not deviate from $s_n^*$ for given $s_{-n}^*$. In this case, such a strategy is open-loop as each miner does not react to others' actions in real time as argued in the previous section. Under the mild assumption of piecewise constant controls of mining rigs, the hash supply functions are simply step functions because it is optimal to turn on the machines in the long run (Lemma~\ref{lem:long_term_s}). And it turns out that a set of controls in Lemma~\ref{lem:simple_s} makes a Nash equilibrium with specific optimal start times. 

\begin{thm}\label{thm:nash}
If there is a Nash equilibrium with step function $s_n^*$'s, then each $s_n^*$ is of form $\gamma_n {\bf 1}_{\{ t \geq t_n^*\}}$ for some $t_n^*$. Conversely, a set of controls specified as such is a Nash equilibrium as long as each $t_n^*$ is the unique solution to $\sigma_n(t_n^*) = 0$. 
\end{thm}
This result provides important insights into the nature of Nash equilibrium. Firstly, previous findings of \cite{tsabary2018gap} are indeed an equilibrium behavior. Secondly, we have a practical tool for computing this equilibrium by solving $\sigma_n(t_n^*) = 0$ recursively. For instance, it is straightforward to verify that the condition $\sigma_N(t_N^*) = 0$ is equivalent to
$$
\int_{t_N^*}^\infty \left( R'(t) - (w R(t) - c) \sum_{i=1}^{N-1} \gamma_i \right) e^{- w \sum_{i=1}^N \gamma_i t} \rd t = 0.
$$
For given $t_{n+1}^*, \ldots, t_N^*$, the condition $\sigma_n(t_n^*) = 0$ can be utilized to infer $t_n^*$. This sequential approach simplifies the complex problem of finding its Nash equilibrium, making it computationally tractable. It also means that when a miner considers its Nash equilibrium, it only needs to consider miners that are larger than itself.

Recall that mining gap occurs, by definition, when computing powers in the blockchain system are partially utilized. Theorem~\ref{thm:nash} together with Lemma~\ref{lem:simple_s} imply that mining gap exists if and only if it is optimal for miner $N$ to delay operations. This is succinctly represented in the next theorem.

\begin{thm} \label{thm:gap_condition}
    Suppose that piecewise constant controls are adopted by all miners. The necessary and sufficient condition for the non-existence of a mining gap at equilibrium is given by 
    \begin{equation}\label{eq:gap_condition}
        R(0) \geq cw^{-1}+ \frac{\max_n \gamma_n}{\sum_{i=1}^N \gamma_i - \max_n \gamma_n}  \bbE \left[R(\bar B)-R(0)\right].
    \end{equation}
Here $\bar B$ follows the exponential distribution with rate $w \sum_{i=1}^N \gamma_i$. 
\end{thm}
This provides an intuitive understanding of the conditions under which a mining gap can occur. The left-hand side of the equation represents the sum of the block reward and the average fee of transactions remaining right after block generation. The first term on the right-hand side is the expected total operating cost because it takes $w^{-1}$ number of hash operations on average for a new block generation, whereas $c$ is the operating cost per unit operation. The term $\bbE[R(B) - R(0)]$ stands for the expected accumulated transaction fees, and the coefficient $\max_n \gamma_n / (\sum_{i=1}^N \gamma_i - \max_n \gamma_n)$ can be interpreted as the odds of the largest miner when all miners are mining diligently. Therefore, if the odds are greater for the largest miner, it requires a larger incentive to make him work. 

The above result outlines three possible scenarios that could lead to a mining gap. The first scenario is when $cw^{-1}$ increases. If the operating cost $c$ or the mining difficulty $w^{-1}$ increases, then the cost of mining may become too large unless there is a large amount of incentive $R(0)$. The second scenario is when the largest miner has a significant dominance. If a single miner controls a substantial portion of the total mining power, he can afford to delay his mining operations for larger revenues, knowing that there is still a high probability of winning. The third and last scenario presents one intriguing and counterintuitive case, namely, the accumulation of large transaction fees during block generation may cause a mining gap. This is due to the strategic decision of miners for more profits, as they may choose to wait for greater rewards. In other words, the mining gap may still occur even when users are willing to pay higher transaction fees. 

\subsection{Homogeneous Network with Affine Revenue}\label{subsec:homogeneous}
Analysis of real Bitcoin data suggests that the reward function $R(t)$ can be approximated linearly, supporting the practical relevance of an affine revenue model. Additionally, current hash rate distribution data indicates that the top six mining pools each hold a relatively uniform share between 10-20\%, collectively accounting for 85\% of the total hash rate. This distribution justifies analyzing these top miners as homogeneous entities, making our simplified model both theoretically interesting and practically applicable.

\begin{cor}\label{cor:homogeneous}
        Suppose that all miners are identical with capacities $\gamma_n = \gamma / N$. Further suppose that the revenue function is given by $R(t) = \alpha t + \beta$ for some constants $\alpha, \beta$. 
Then, the optimal starting time $t^*$ at equilibrum in Theorem~\ref{thm:nash} is given by $t^* = ((cw^{-1} - \beta)/\alpha + 1/((N-1)\gamma w) )^+$. As a result, mining gap does not occur if and only if $\beta \geq c w^{-1} + \alpha / ((N-1) w \gamma)$. 
\end{cor}
    
\begin{proof}
    Under the given assumptions, the switching function $\sigma(t)$ can be reformulated as follows: for $t \geq t^*$ and for any $n$, 
    \begin{eqnarray*}
            \sigma_n(t) &=& w \int_{t}^\infty \left(\left(w(\alpha u+\beta)-c\right)\frac{N-1}{N} \gamma-\alpha\right)e^{-w\gamma(u-t^*)}\rd u\\
            &=& \frac{e^{-w \gamma (t - t^*)} }{\gamma} \left[ \frac{N-1}{N} \left( \alpha + w \gamma \alpha t + w \gamma \beta - c \gamma\right) - \alpha \right]. 
    \end{eqnarray*}
    By solving $\sigma_N(t^*) = 0$ and finding a condition for $t^* \geq 0$, we obtain the desired result. \qed 
\end{proof}
Corollary~\ref{cor:homogeneous}'s condition for the absence of a mining gap, $\beta \geq c w^{-1} + \alpha / ((N-1) w \gamma)$, provides valuable insights into miners' strategic behavior. Here, $1/w\gamma$ represents the expected block time when there is no mining gap. The term $1/(N-1)$ represents the odds of winning for each miner. Thus, $\alpha / ((N-1) w \gamma)$ can be interpreted as the relevant expected revenue a miner can obtain by participating in this game.

The homogeneous miners with affine $R$ allow us not only to obtain simple $t^*$ but also to compute the expected block generation time and the utility of miners in closed form. Firstly, the expected block generation time at equilibrium is given by 
$$
\bbE[B] = \bbE[\bar B + t^*] =  \frac{1}{\gamma w}  + t^* = \max\left\{ \frac{1}{\gamma w} , \frac{cw^{-1} - \beta}{\alpha} + \frac{N}{(N-1)\gamma  w}\right\} 
$$
where $\bar B$ is the exponential random variable with rate $\gamma w$. 
As for the utility at equilibrium, we proceed as follows:
\begin{eqnarray*}
J_n^* &=&  \int_{t^*}^\infty (w R(t) - c ) \frac{\gamma}{N} e^{- \gamma( t - t^*)} \rd t \\
&=& \frac{1}{N} \bbE\left[ R(\bar B + t^*) - \frac{c}{w}\right] \\
&=& \max\left\{\frac{1}{N}\left(\frac{\alpha}{\gamma w}+\beta-cw^{-1}\right), \frac{1}{N-1}\frac{\alpha}{\gamma w} \right\}. 
\end{eqnarray*}
The utility at equilibrium, $J_n^*$, never falls below $\alpha / ((N-1)\gamma w)$, which is independent of $\beta$ and $c$. Since $\beta$ contains the block reward and $c$ is the production cost, this observation implies that the mining gap is a strategic response by miners to maintain this minimum level of revenue regardless of these two system components. 

Before we end this section, we consider one important system component, namely the target block generation time $\mu$. As explained earlier, the difficulty adjustment algorithm adjusts the winning rate $w$ to achieve this goal. For the simplified setting of this subsection, it is possible to compute the exact difficulty level that makes this possible. This effectively links strategic decisions of miners to the system's stability. More details are given in the next section. 

\begin{cor}\label{cor:homogeneous_w}
 Suppose that all miners are identical with capacities $\gamma_n = \gamma / N$. Further suppose that the revenue function is given by $R(t) = \alpha t + \beta$ for some constants $\alpha, \beta$. 
Then, the target block generation time $\mu$ is achieved by $w = (\gamma \mu)^{-1}$ without mining gap at equilibrium in Theorem~\ref{thm:nash} if $ \beta \geq \Delta \gamma \mu$. Otherwise, $w = (\Delta + \alpha \gamma^{-1})/(\alpha \mu + \beta)$ and the optimal delay $t^* = (\Delta \mu \gamma  - \beta) / (\Delta \gamma + \alpha)$ where $\Delta = c + \alpha \gamma^{-1} / (N-1)$. 
\end{cor}

\begin{proof}
    First, consider $\beta \geq \Delta \gamma \mu$. We can verify that there is no mining gap if $w = (\gamma \mu)^{-1}$, satisfying the condition $\beta \geq \Delta / w$ in Corollary~\ref{cor:homogeneous}. We claim that no $w$ can incur a mining gap in this case. If there were a difficulty $w$ achieving the target generation time $\mu$ while satisfying $\beta < \Delta / w$, then by Corollary~\ref{cor:homogeneous}, $w = (\Delta + \alpha \gamma^{-1})/(\alpha \mu + \beta)$. This, combined with $\beta < \Delta / w$, would imply $\beta < \Delta \gamma \mu$, contradicting our initial assumption.
    
    Now, suppose $\beta < \Delta \gamma \mu$. In this case, no difficulty $w$ can achieve the target $\mu$ without a mining gap. If such a $w$ existed, it would be $w = (\gamma \mu)^{-1}$, leading to $\beta \geq \Delta w^{-1}$, which contradicts our assumption. Therefore, $w$ must be set considering a positive delay $t^*$. From the previous corollary, we can directly derive the expressions for $w$ and $t^*$ as given in the statement.
    \qed
\end{proof}

The difficulty at equilibrium above can also be succinctly written as $w = \max\{ (\gamma \mu)^{-1}, (\Delta + \alpha \gamma^{-1})/(\alpha \mu + \beta)\}$ and $t^* = (\Delta \mu \gamma  - \beta)^+ / (\Delta \gamma + \alpha)$. This result illustrates that miners cease to mine if the system adheres to the fixed target time $\mu$ under insufficient revenue $R(0)$. The right hand side of the critical condition $\beta \geq \Delta \gamma \mu$ in a very large $N$ case is nothing but the operating cost of diligent mining by all miners for the duration of a single block generation. 
Reduced difficulty or higher winning probability as an effort to meet the target $\mu$ when $\beta < \Delta \gamma \mu$ might look advantageous to individual miners. However, it may still induce non-utilization of the computing power of the network and make the system vulnerable to attacks. This phenomenon is not merely an artifact of this simplified setting but a systematic characteristic caused by the trade-off between operating costs and accumulated rewards.

\section{Difficulty Adjustment and Mining Gap}\label{sec:DAA}

\subsection{The Protocol}

Recall that the revenue function $R(\cdot)$ and the cost per hash $c$ are exogenously given, but  the winning rate $w$ is regularly updated by the network protocol. 
This updating method DAA for the case of Bitcoin uses historical block generation times and recalibrates difficulty every 2016 blocks (approximately biweekly). The adjusted winning rate is calculated as
$$
w_{\rm new} = \bw\left( \mu_{\rm old}, w_{\rm old}\right) =  \max\left(\frac{1}{4},\min\left(\frac{\mu_{\rm old}}{\mu},4\right)\right) \cdot w_{\rm old}.
$$
Here, $w_{\rm old}$ is the previous winning rate, $w_{\rm new}$ is the new rate, $\mu_{\rm old}$ is the average of actual block generation times, and $\mu$ is the target time. The formula adjusts the winning rate to maintain a constant product of $w_{\rm old}$ and $\mu_{\rm old}$.Therefore, if $\mu_{\rm old}$ is less (more) than $\mu$, then the protocol decreases (increases) the winning probability so as to make winning more (less) difficult. However, such a change in the winning rate is set to be within 25\% and 400\% range.

For analytical tractability and also based on a rationale that deviating beyond such a range is highly unlikely, we make the following assumption. 

\begin{ass}\label{assumption:DAA}
The winning rate is adjusted by $w_{\rm new} = \bw(\mu_{\rm old}, w_{\rm old}) = \mu_{\rm old} w_{\rm old} / \mu$. 
\end{ass}
New block generation time $\mu_{\rm new}$ is tied to $w_{\rm new}$ and is a function of miners' optimal responses:
\begin{equation}\label{eq:bmu}
\mu_{\rm new} = \bmu(w_{\rm new}) = \int_0^\infty \exp\left( - w_{\rm new} x^*(t) \right) \rd t.
\end{equation}
This DAA results in iterative sequences of winning rates and average block generation times. A successful DAA would yield $\mu = \bmu( w)$ for the target $\mu$ and a suitable winning rate $w$.

\subsection{System Stability}
We denote the outcome of repeated DAA applications by $w_i$ and $\mu_i$, where $\mu_i$ is the average block generation time associated with $w_i$. One natural question is then whether the algorithm actually does what it is designed to achieve. To determine if the algorithm achieves its intended goal, we define steady-state winning rate and $\varepsilon$-stable winning rate.
 
    \begin{defn}
        A winning rate $w$ is in \emph{steady-state} if $\mu = \bmu (w)$. 
    \end{defn}

\begin{defn}
A steady-state winning rate $w$ is said to be \emph{$\varepsilon$-stable} if there exists some positive $\varepsilon$ such that $\lim_{i \rightarrow \infty} w_i = w$ where $w_i = \bw\left( \mu_{i-1}, w_{i-1}\right)$ and $\mu_i = \bmu( w_{i})$ for any initial $w_0$ with $|w_0 - w| < \varepsilon$. 
\end{defn}

Corollary~\ref{cor:homogeneous_w} shows the existence of steady-state $w$ depending on $R$. Even with mining gaps, it's possible to choose $w$ satisfying $\mu = \bmu(w)$. The following theorem proves that steady-state $w$ is guaranteed for the Nash equilibrium in Section~\ref{subsec:nash}.

    \begin{thm}\label{thm:steady_state}
        There exists a steady-state winning rate $w$ at equilibrium in Theorem~\ref{thm:nash}. 
    \end{thm}

In current Bitcoin networks, the difficulty is physically limited to between $2^{-32}$ and $2^{-256}$.\footnote{In each Bitcoin block, the hash result must fall below a certain target. The highest possible target, set at 0x1d00ffff (roughly $2^{224}$), is embedded directly in the Bitcoin source code. Since the number of all possible hash results is $2^{256}$, this corresponds to a maximum winning probability per hash of $2^{-32}$. On the other hand, the minimum theoretical winning probability is $2^{-256}$.} However, given today's typical computing power, it can be treated as minuscule. Data from March 2024 shows that the Bitcoin network performs $2^{78}$ hashes in 10 minutes, and a single Antminer KS5 performs $2^{53}$ hashes in the same amount of time. Therefore, the physical constraints cover the most extreme cases and are a negligible constraint in our model. However, the existence of a steady state doesn't guarantee stability. Due to fluctuations in $R$ and $c$, the equilibrium system can easily deviate.

Consider the sequence of winning rates $w_i$ and mean block generation times $\mu_i$. The update protocol in Assumption~\ref{assumption:DAA} can be rewritten as
$$
x_{i+1} = f(x_i) + x_i
$$
where $x_i = \ln (w_i/w)$ and $f(x) = \ln \left(\bmu\left(w e^x \right)/\mu \right)$. The convergence of ${w_i}$ to $w$ is equivalent to $\lim_{ i \rightarrow \infty} x_i = 0$. When it comes to the convergence of real sequences, there are some well known elementary convergence criteria. One instance is the existence of a constant $l$ such that $| x_{i+1} / x_i | \leq l < 1$ for all $i$'s. This can be translated into the following simple condition.

\begin{lem}\label{lem:DAA_conv}
Suppose $|f(x)/x + 1| \leq l  $ for some constant $l < 1$ whenever $| x | < \varepsilon_0$ for some $\varepsilon_0$. Then, the given steady-state winning rate $w$ is $\varepsilon$-stable for some suitable $\varepsilon$. 
\end{lem}
The above condition is essentially about the variation in $\bmu$ when the winning rate changes. More specifically, straightforward calculations show that it is equivalent to $\mu / \delta^{1+l} \leq \bmu( \delta w) \leq \mu/ \delta^{1-l}$ for the case of $\delta  > 1$. Similar inequalities are obtained for $\delta  < 1$. 

A graphical illustration of $f(\cdot)$ is given in Figure~\ref{fig:2.4.DAA_conv}. For given $x$ or winning rate $w e^x$, $\bmu(w e^x) \geq (\gamma w)^{-1} e^{-x}$ where equality holds when there is no mining gap. Therefore, $f(x) \geq - x - \ln (\gamma w \mu)$. Also, the equality holds for sufficiently large $x$ values. Now suppose initial $x_0$ or $w_0 = w e^{x_0}$ is given. Then, the next value $x_1$ is obtained by the zero of the function $y = - x + x_0 + f(x_0)$. Sequentially, $x_i$'s on the $x$-axis are determined. The figure shows one such convergent case of $x_i$'s to the origin. 

    \begin{figure}[htb!]
        \centering
        \includegraphics[width=0.5\textwidth]{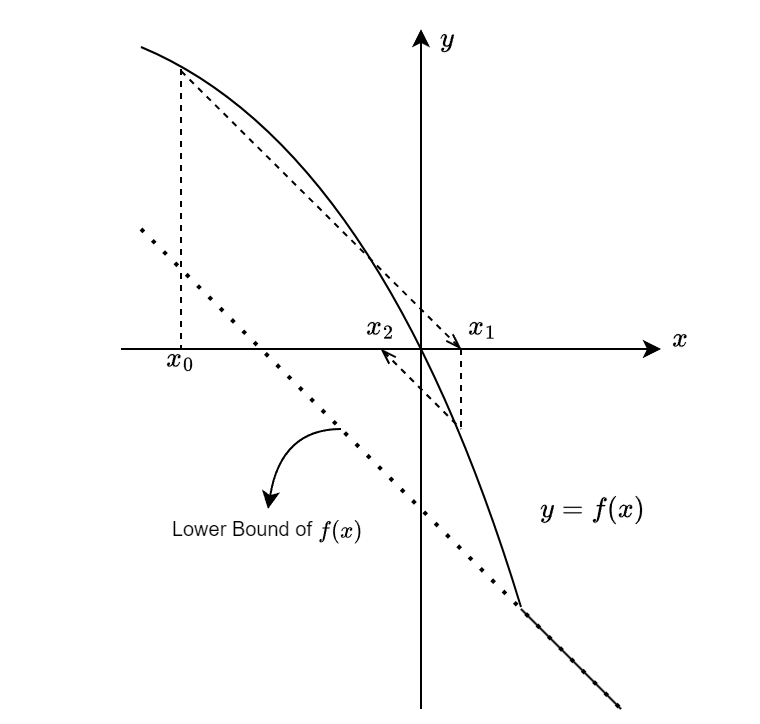}
        \caption{Example of difficulty update path}
        \label{fig:2.4.DAA_conv}
    \end{figure}

 While Lemma~\ref{lem:DAA_conv} is elementary and the underlying idea is applicable to any real sequence, the following result uses the characteristics of $\bmu$ more heavily. We utilize the differentiability of $\bmu$ as argued in the proof of Theorem~\ref{thm:steady_state}. 
 
\begin{thm}\label{thm:DAA_conv}
Suppose $w$ is a steady-state winning rate and $\mu$ is the target block generation time. If $\bmu'(w) >  - 2 \mu / w$, then $w$ is $\varepsilon$-stable for some suitable $\varepsilon$. If $\bmu'(w) < -2 \mu / w$, then there is no such $\varepsilon$. 
\end{thm}
%
Studies like \cite{goren2019mind} have shown that miners can manipulate Bitcoin's DAA for profit. Theorem~\ref{thm:DAA_conv} suggests that even without malicious intent, similar instability can occur in mining gap situations if the stated condition is met. When unstable, block times may oscillate wildly, severely impacting network usability and predictability. Such volatility not only complicates Bitcoin transactions but also undermines system reliability and user trust. The potential for network destabilization and increased vulnerability to attacks makes this an undesirable scenario that should be avoided. 

\subsection{Homogeneous Network with Affine Revenue}
The homogeneous case in this subsection allows for explicit system stability requirements. For a steady-state winning rate $w$ without mining gap, $w \bmu'(w) / \mu = -1$, ensuring $\varepsilon$-stability by Theorem~\ref{thm:DAA_conv}. The following corollary addresses the mining gap case.

 \begin{cor}\label{cor:DAA_conv}
 Suppose that all miners are identical with capacities $\gamma_n = \gamma/ N$. Further suppose that the revenue function is given by $R(t) = \alpha t + \beta$ for some constants $\alpha$, $\beta$. Let $w$ be a given steady-state winning rate with a mining gap, and $\mu$ be the target block generation time. If $\beta  < \alpha \mu$, then $w$ is $\varepsilon$-stable for some suitable $\varepsilon$. If $\beta > \alpha \mu$, then there is no such $\varepsilon$. 
 \end{cor}
 
 \begin{proof}
 Recall that we computed $\bbE[B]$ in Section~\ref{subsec:homogeneous}. When there is a mining gap, we have
 $$
 \bmu(w e^x) = \frac{c w^{-1} e^{-x} - \beta}{\alpha} + \frac{N}{(N-1)\gamma w} e^{-x}
 $$
 for all sufficiently small $x$. In particular, $c/\alpha + N/((N-1)\gamma)  = (\mu + \beta/ \alpha) w$ from $\bmu(w) = \mu$. 
  Straightforward calculations yield $w \bmu'(w) / \mu = -1 - \beta/(\alpha \mu)$. This is greater than $-2$ if and only if $\beta < \alpha \mu$. The case of no stability is also a direct consequence of Theorem~\ref{thm:DAA_conv}. 
\qed
 \end{proof}

Some appropriate size of $\varepsilon$ in the statement of the above corollary can be more explicit by using criteria such as Lemma~\ref{lem:DAA_conv}, although it is not pursued further to economize on space. One intriguing implication of the above result is that a larger block reward (or higher $\beta$) with steady-state winning rate may make the system stability not achievable. Or, in other words, if the accumulated fee per unit time is insufficient relative to $\beta$, then the DAA may fail to accomplish the very objective that it is designed. 
In this simplified setting, we can then divide the whole region of $\alpha$ and $\beta$ into three. The first is the case in which there is no mining gap thanks to the condition $\beta \geq \Delta \gamma \mu$ in Corollary~\ref{cor:homogeneous}. The second is the case in which there is a mining gap but the given steady-state winning rate in the same corollary is $\varepsilon$-stable for some $\varepsilon$; $\beta < \min\{\Delta\gamma\mu, \alpha \mu\}$. The last is the case of a mining gap with no system stability: $\alpha \mu < \beta < \Delta\gamma\mu$. 

We end this section with some simple numerical tests. Let us assume that $\mu = 1$, $N = 10$, $\gamma = 1$, $c = 2$, and $R(t) = 2t+3$. The operating cost rate is intentionally set high to induce the situation of a mining gap. Corollary~\ref{cor:homogeneous} implies that there are a unique steady-state winning rate $w = 22/9$ and a mining gap $t^* = 13/22$. Corollary~\ref{cor:DAA_conv} tells us that system stability is not achievable for this winning rate. 
The upper left panel of Figure~\ref{fig:DAA_conv} shows the graph of function $f(x)$ and the straight line $y = -x - \ln (\gamma w \mu)$. The upper right panel of the figure then shows two sample paths when the update algorithm starts at two different $x_0$. Rather than converging, the DAA creates sequences of $x_i$'s that oscillate. In terms of the winning rate, see the lower left panel of the figure. 
Combining Corollaries~\ref{cor:homogeneous} and ~\ref{cor:DAA_conv}, 
We can find three possible regions in the $\alpha\mu$-$\beta$ plane, as in the bottom right panel of the figure. The point $(2,3)$, corresponding to $(\alpha\mu, \beta)$ of the revenue function $R(t)=2t+3$, is located in the orange-colored region, and any point in the region means that the relevant system produces a mining gap and there is no stable steady-state winning rate. If, however, the point moves to the blue-colored region, then the system may still have a mining gap but a stable steady-state winning rate exists. As long as $\beta$ is large enough (white part), there is no mining gap. 

\begin{figure}[t]
     \centering
     \begin{subfigure}{0.48\textwidth}
         \centering
         \includegraphics[width=\textwidth]{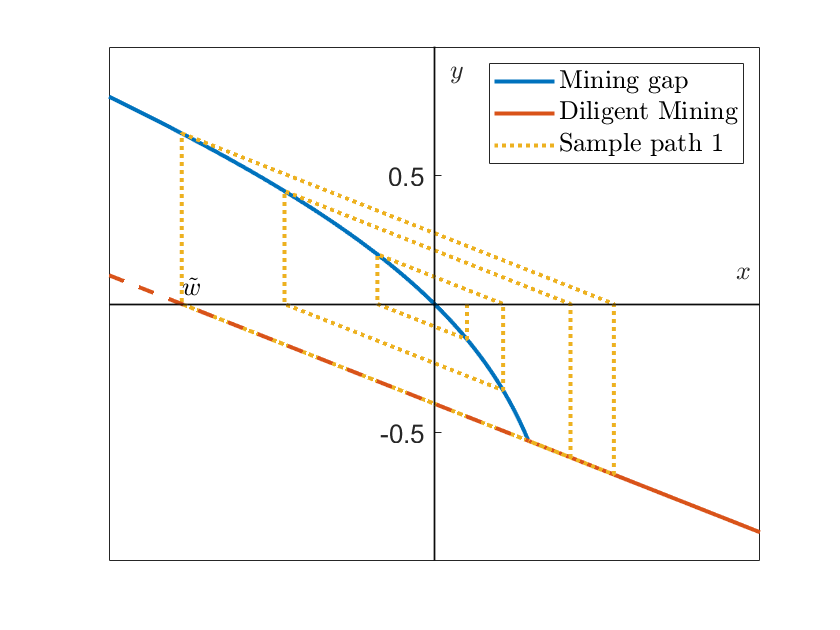}
         \caption{Non-convergent sample path 1}
     \end{subfigure}
     \hfill
     \begin{subfigure}{0.48\textwidth}
         \centering
         \includegraphics[width=\textwidth]{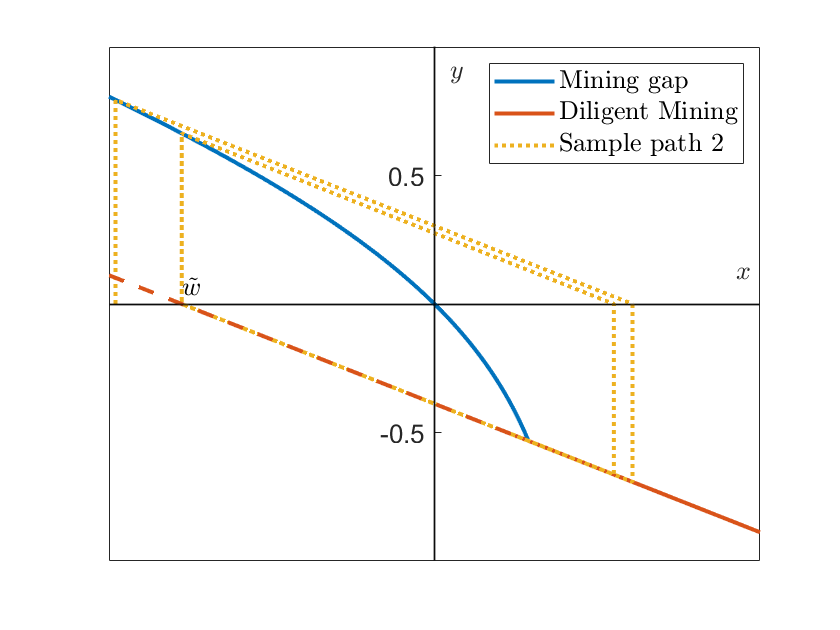}
         \caption{Non-convergent sample path 2}
     \end{subfigure}     
     \hfill
     \begin{subfigure}{0.48\textwidth}
         \centering
         \includegraphics[width=\textwidth]{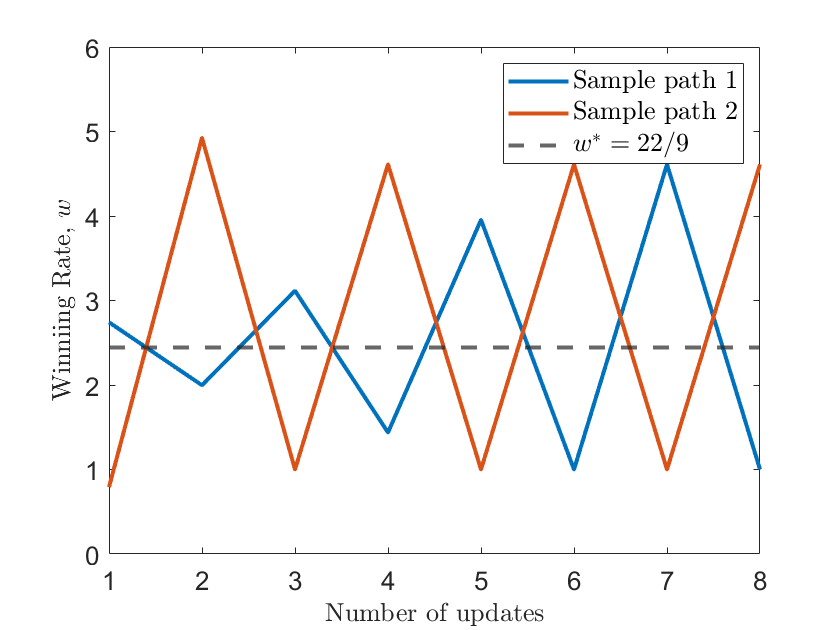}
         \caption{Difficulties of sample path}
     \end{subfigure}
     \hfill
     \begin{subfigure}{0.48\textwidth}
         \centering
         \includegraphics[width=\textwidth]{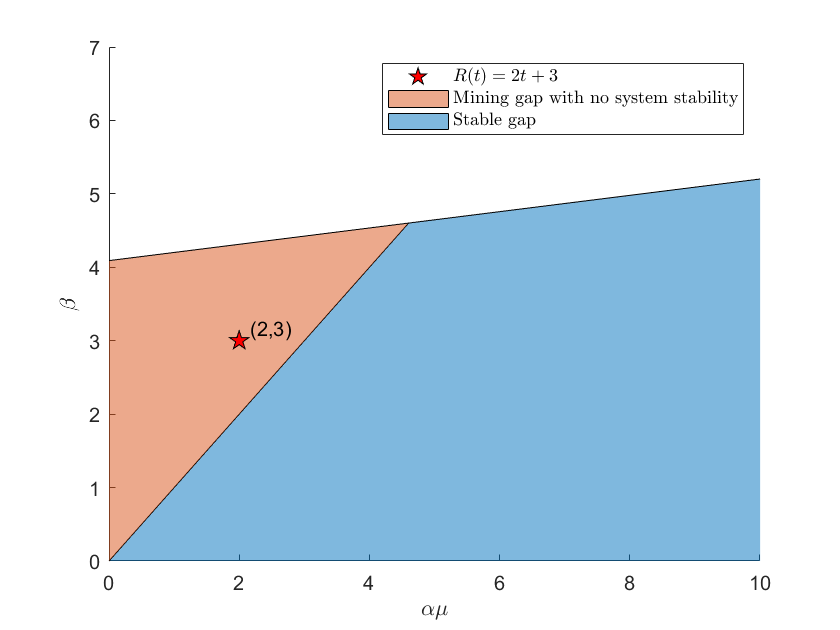}
         \caption{$\alpha\mu$-$\beta$ plane and three regions}
     \end{subfigure}
     \caption{Interaction between DAA and mining gap}
     \label{fig:DAA_conv}
\end{figure}

\section{Further Considerations}\label{sec:Simulation}
This section examines the features of the target system numerically, incorporating  miners' strategic decisions and the difficulty adjustment protocol. We focus on a two-player case. While such a scenario is unlikely in real-world Proof-of-Work systems, and a dominant player with over 50\% hash power would likely pursue more direct manipulation strategies, we adopt this simplified model for its analytical tractability and for the ease of visualization. Our focus is solely on the mining gap phenomenon and its implications.

\subsection{Sensitivity Analysis}

{\bf Optimal Response.} We consider three revenue functions: (1) linear $R(t) = t+2$, (2) square root $R(t) = 1.13\sqrt{t}+2$, and (3) log $R(t) = 5.03 \log(t+1)$. The mean block generation time without mining gap, $(\gamma w)^{-1}$, is set to $\bbE[\bar B]=1$. All revenue functions yield an expected revenue of $\bbE[R(\bar B)]= 3$. We assign 30\% of the computing power to the small miner ($\gamma_1= 0.3 \gamma$) and 70\% to the large miner ($\gamma_2 = 0.7 \gamma$). The operating cost rate and winning rate are $c = 2\gamma^{-1}$ and $w = \gamma^{-1}$, respectively.

Figure~\ref{subfig:R1} shows the graphs of three different revenue functions together with the mean block generation time. On the right are there optimal responses of the two miners. As noted in Lemma~\ref{lem:simple_s}, the small miner activates earlier in all cases. In addition, depending on the shape of the revenue function, the mining gap is smaller if the initial reward is larger (linear versus log). And in the case of the same initial reward (linear versus square-root), the mining gap is smaller if the growth of the revenue is steeper at the beginning. 

    \begin{figure}[htb!]
     \centering
     \begin{subfigure}{0.48\textwidth}
        \centering
         \includegraphics[width=\textwidth]{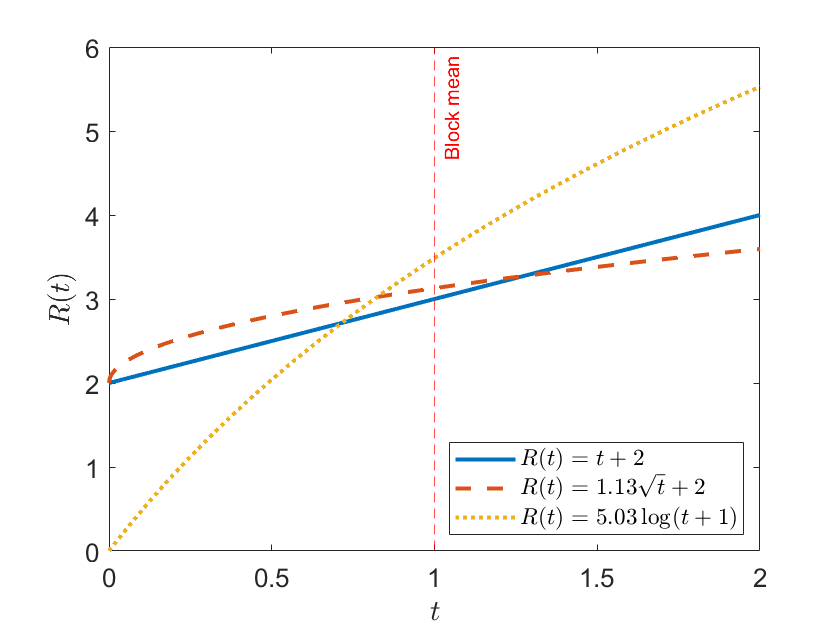}
         \caption{Three different revenue functions.}
         \label{subfig:R1}
     \end{subfigure}
     \hfill
     \begin{subfigure}{0.48\textwidth}
        \centering
         \includegraphics[width=\textwidth]{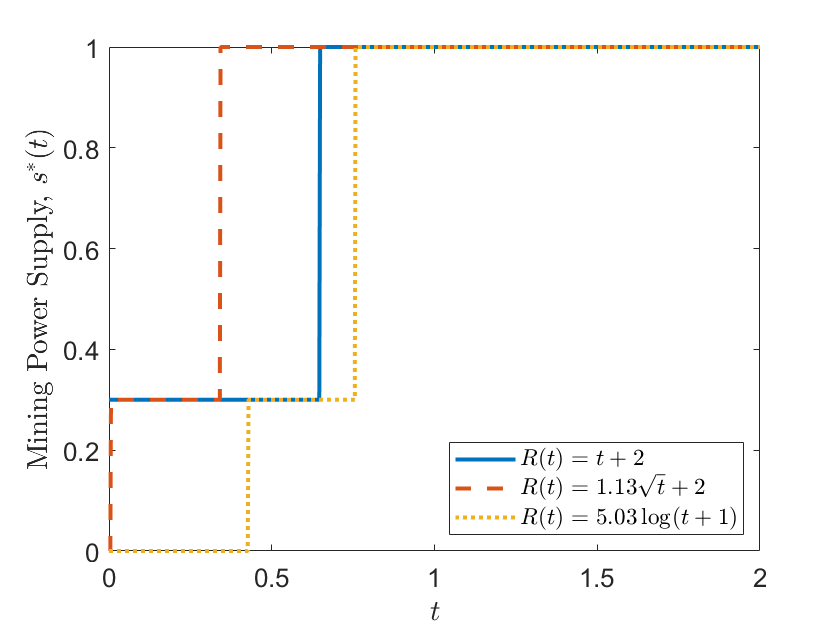}
         \caption{Optimal responses of miners}
         \label{subfig:R2}
     \end{subfigure}
     \caption{The optimal responses of the two miners under three types of revenue functions}
    \label{fig:2.5.R comparision}
    \end{figure}

\vs 
\noindent 
{\bf Steady-state Analysis.} For given parameters and an initial winning rate $w_0$, we apply the DAA to generate a sequence of winning rates $w_i$ and check for convergence to a steady state. Our numerical example above is designed to induce this convergence of winning rates. As a result, for a given value of $c$ and the revenue function $R(t)$, we obtain the optimal responses of the two miners at the Nash equilibrium in the steady-state. Using $R(t) = t+2$, we vary dominance of the large miner from $\gamma_2 = 0.5 \gamma$ to $0.95 \gamma$, with $c = 2\gamma^{-1}$. 

Figure~\ref{subfig:Dom3} shows the steady-state winning rate increasing with the large miner's dominance. Figure~\ref{subfig:Dom1} demonstrates that despite higher winning rates, the large miner delays activation while the small miner starts earlier. Figure~\ref{subfig:Dom2} shows relatively stable optimal utilities for both miners across dominance levels. This is the combined effect of higher winning rates and the distinct optimal responses of the two miners. 
The lower right panel, Figure~\ref{subfig:Dom4}, provides some additional information. For each dominance level of the large miner, we find the steady-state Nash equilibrium, and then use that to calculate the following equation to get the utilization level of the total computing power per block generation. Mathematically, utilization is defined as $\gamma^{-1} \int_0^\infty s^*(t) f_B(t)\rd t$. 
The total utilization of mining rigs is shown to decrease as the dominance level increases. And it is only around 60\% at $\gamma_2 = 0.5 \gamma$ and drops below $0.1 \gamma$ if the dominance is much greater. This finding resonates with some of the previous findings in the literature on low utilization of system resources in blockchains, as shown in \cite{tsabary2018gap}.

\begin{figure}[htb!]
     \centering
     \begin{subfigure}{0.48\textwidth}
         \centering
         \includegraphics[width=\textwidth]{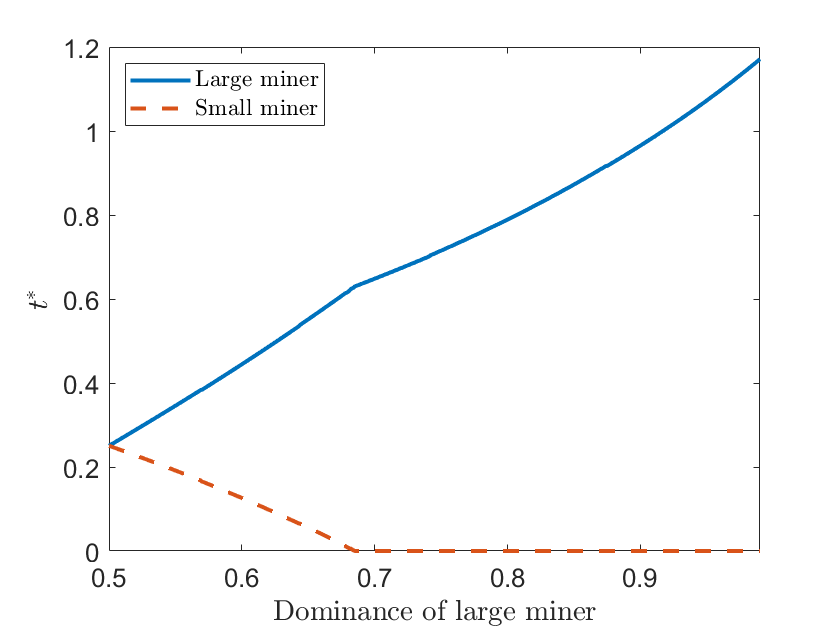}
         \caption{Optimal starting   time}
         \label{subfig:Dom1}
     \end{subfigure}
     \hfill
     \begin{subfigure}{0.48\textwidth}
         \centering
         \includegraphics[width=\textwidth]{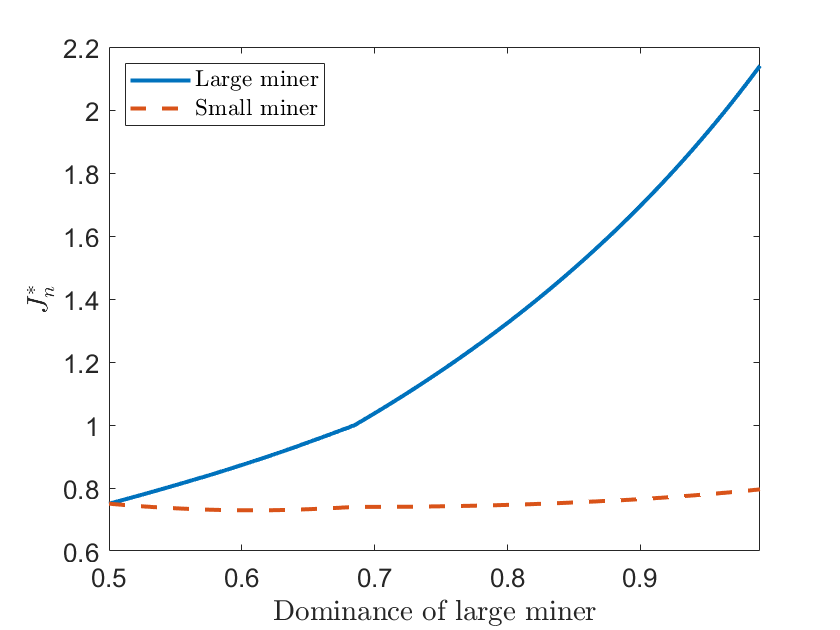}
         \caption{Utility of miners}
         \label{subfig:Dom2}
     \end{subfigure}
     \hfill
     \begin{subfigure}{0.48\textwidth}
         \centering
         \includegraphics[width=\textwidth]{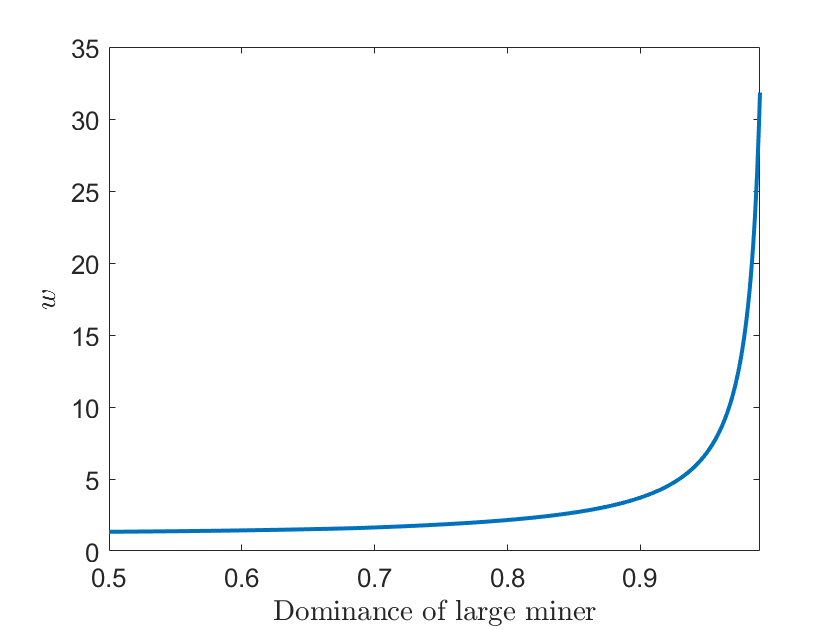}
         \caption{Difficulty}
         \label{subfig:Dom3}
     \end{subfigure}
     \hfill
     \begin{subfigure}{0.48\textwidth}
         \centering
         \includegraphics[width=\textwidth]{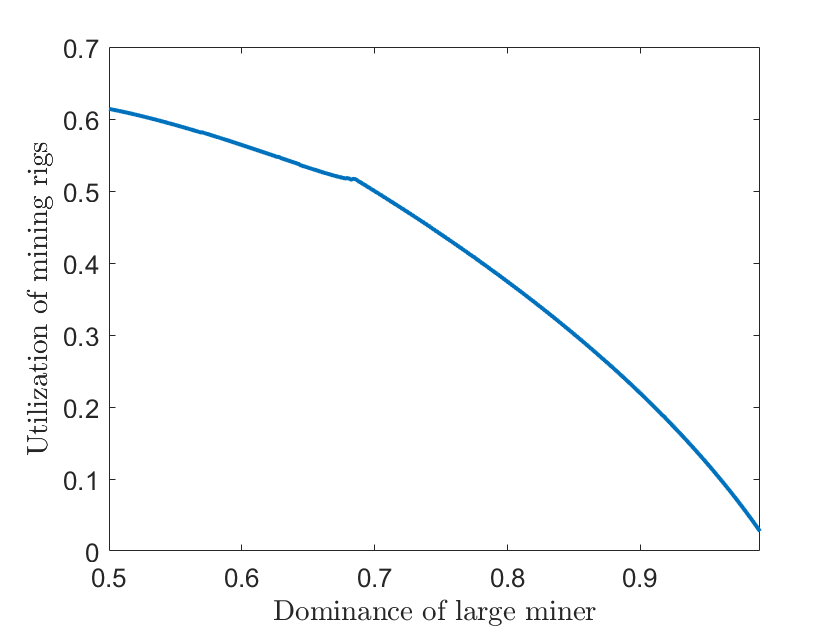}
         \caption{Utilization of mining rigs}
         \label{subfig:Dom4}
     \end{subfigure}
          \hfill
     \caption{Behaviors of miners as the dominance of the large miner varies. At each dominance level, optimal responses are computed and the system is run until the sequence of difficulties converges.}
     \label{fig:2.5.Dominance}
\end{figure}
\vs
\noindent
{\bf Effect of Cost.} We examine the effect of varying operating cost rate c, fixing the large miner's dominance at 70\% and using $R(t) = t+2$. Again, the DAA is applied until the sequence of winning rates converges. The simulation outcomes shown in the four panels of Figure~\ref{fig:2.5.cost} are as expected. A higher cost rate means decreased profits and thus the optimal responses of the miners are to delay operation while winning rates are settled higher. Nevertheless, the decrease in the optimal utilities is inevitable. Under-utilization of the computing power prevails as well.

\begin{figure}[htb!]
     \centering
     \begin{subfigure}{0.48\textwidth}
         \centering
         \includegraphics[width=\textwidth]{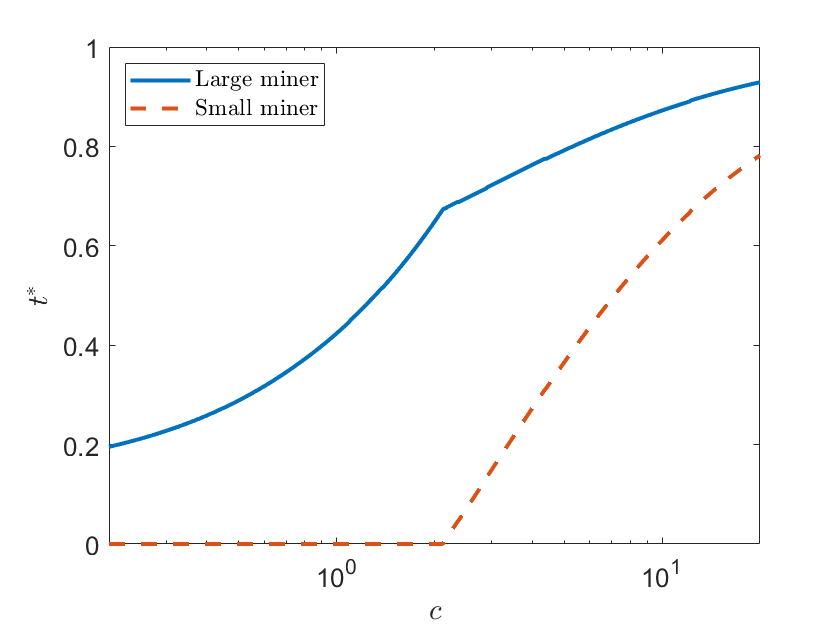}
         \caption{Optimal starting time}
         \label{subfig:cost1}
     \end{subfigure}
     \hfill
     \begin{subfigure}{0.48\textwidth}
         \centering
         \includegraphics[width=\textwidth]{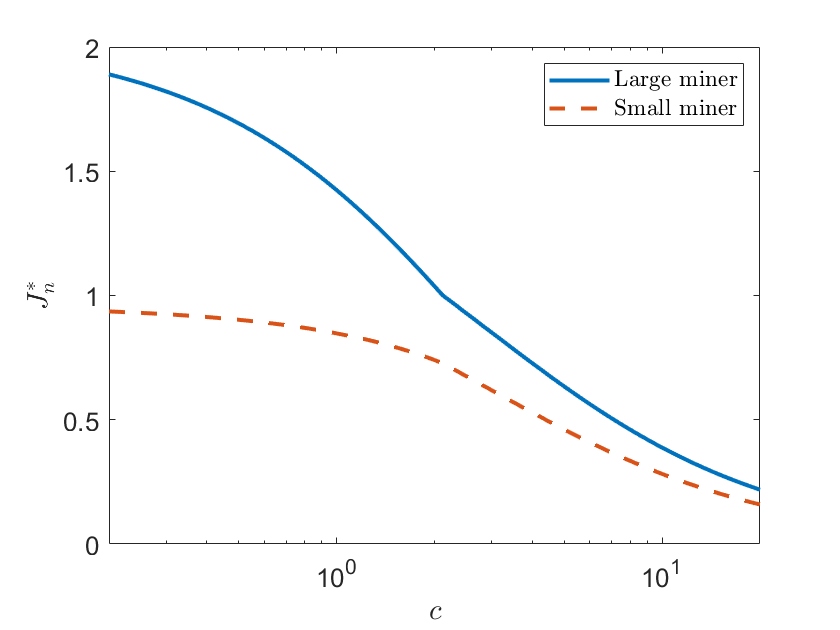}
         \caption{Utility of miners}
         \label{subfig:cost2}
     \end{subfigure}
     \hfill
     \begin{subfigure}{0.48\textwidth}
         \centering
         \includegraphics[width=\textwidth]{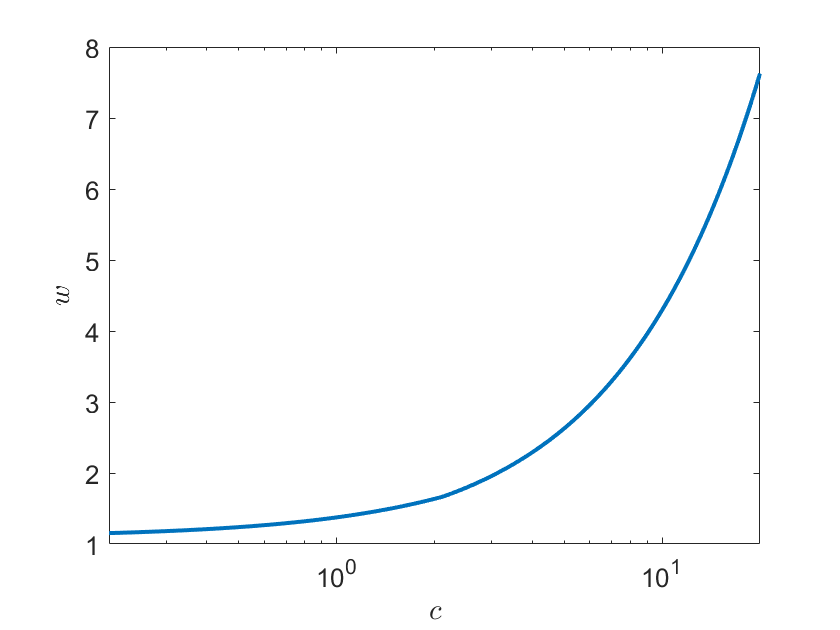}
         \caption{Difficulty}
         \label{subfig:cost3}
     \end{subfigure}
     \hfill
     \begin{subfigure}{0.48\textwidth}
         \centering
         \includegraphics[width=\textwidth]{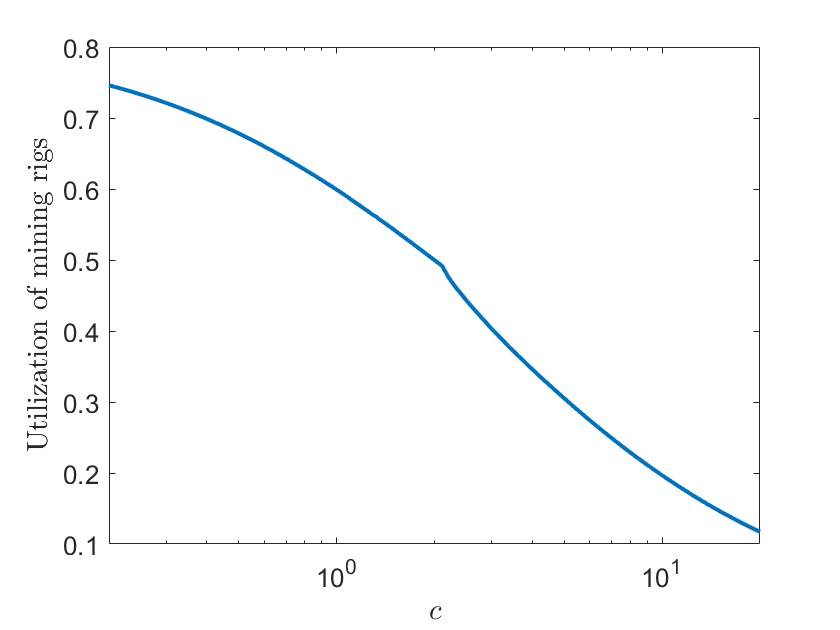}
         \caption{Utilization of mining rigs}
         \label{subfig:cost4}
     \end{subfigure}
     \caption{Behaviors of the two miners as the cost $c$ varies. Naturally, higher costs cause adverse effects for mining profitability.}
     \label{fig:2.5.cost}
\end{figure}

\subsection{Coalition}\label{subsec:coalition}

The issue of centralization, that is, the actual phenomenon that miners form a large mining pool which acts as a single miner, has been discussed in the literature \citep{eyal2015miner, eyal2018majority}. The fundamental principle of blockchain systems is to maintain a decentralized network of miners and users without a single governing entity. Nevertheless, our analysis below reveals that the concentration of hashing power can be profitable due to strategic mining gaps. For instance, in Figure~\ref{subfig:Dom2} is it not difficult to see that the sum of the utilities of large and small miners increase as the dominance of the large miner increases. This is achieved while the utility of the small miner is not much affected.

To see this, let us consider miners $n$ and $m$ in a given network where $t_i^*$'s are optimal responses derived in Theorem~\ref{thm:nash}. Assume that at least one of $t_n^*$ and $t_m^*$ is positive and that miners $n$ and $m$ form a coalesced mining pool, which leads to a new network. Then the utility of this coalition is 
$$
J_{n,m}^* = \max \int (wR(t) - c)s_{n,m}(t) e^{- x(t)} \rd t
$$
where $s_{n,m}(t)$ is the operating strategy of the coalition with the maximal capacity $\gamma_n+\gamma_m$. Since the set of feasible strategies subsumes the strategy of operating $\gamma_n$ at $t_n^*$ and $\gamma_m$ at $t_m^*$, it is obvious $J_{n,m}^* \geq J_n^* + J_m^*$. More interestingly, one can show that optimal starting times of miners in the new network are greater, if not equal, than optimal starting times in the original network. 

Lemma~\ref{lem:nash compare} below is helpful in this regard. Recall that optimal responses increase in the increasing order of capacities in Lemma~\ref{lem:simple_s} and Theorem~\ref{thm:nash}. For miners with capacities greater than $\gamma_n + \gamma_m$, start times are not changed. Let us denote the miner whose capacity is the largest among miners with capacities smaller than $\gamma_n + \gamma_m$ by $\triangle$. It is not difficult to argue that the start time of miner $\triangle$ in the original network is smaller than that of the coalition in the new network, thanks to the lemma below. This in turn implies that the start time of miner $\triangle$ in the new network is greater than the start time of the same miner in the original network. Applying the same argument to the miners in the decreasing order of capacities, we arrive at the desired conclusion.

\begin{lem}
    \label{lem:nash compare}
    Consider a miner with capacity $\gamma_n$ in two different networks of miners. For the respective Nash equilibria $\{s_i^*\}$ and $\{\tilde s_i^*\}$ as in Theorem~\ref{thm:nash}, suppose that miner $n$'s capacity in one network is relatively smaller at all times than the other, say $ s_{-n}^*(\cdot) \leq \tilde s_{-n}^*(\cdot)$, then the optimal response time in the former network is greater, or simply $ t_n^* \geq \tilde t_n^*$. 
\end{lem}

This result implies that a miner would start earlier if the environment is more competitive or if competitors have larger capacities. Figure~\ref{subfig:Dom1} shows this result graphically in the case of two miners. If miners form coalitions instead of installing new capacities, the combined utilities are greater. However, miners choose to delay start times, and there may be significant impacts on the efficiency and resource utilization of the network, potentially leading to increased transaction fees and reduced network throughput. 

Lastly, we would like to point out that a more powerful large miner does not necessarily extract profit from a smaller miner, contrary to intuition. We define the utility per unit capacity as $J_n^*/\gamma_n$. Figure~\ref{subfig:Dom5} illustrates that this utility per capacity increases for the small miner as hashing power concentrates on the large miner. The utility per capacity for the large miner only begins to increase at the point where the small miner operates without gaps (i.e., $t_2^* = 0$), marked by the vertical line in the figure.
This phenomenon can be explained as follows: As the large miner becomes more powerful, larger mining gaps are created, leading to lower difficulties. This benefits the small miner who starts working earlier. The large miner only benefits from improved unit profitability when the difficulty is sufficiently low so that the small miner no longer needs to create a mining gap.

\begin{figure}
\centering
\begin{minipage}{.5\textwidth}
  \centering
  \includegraphics[width=1\linewidth]{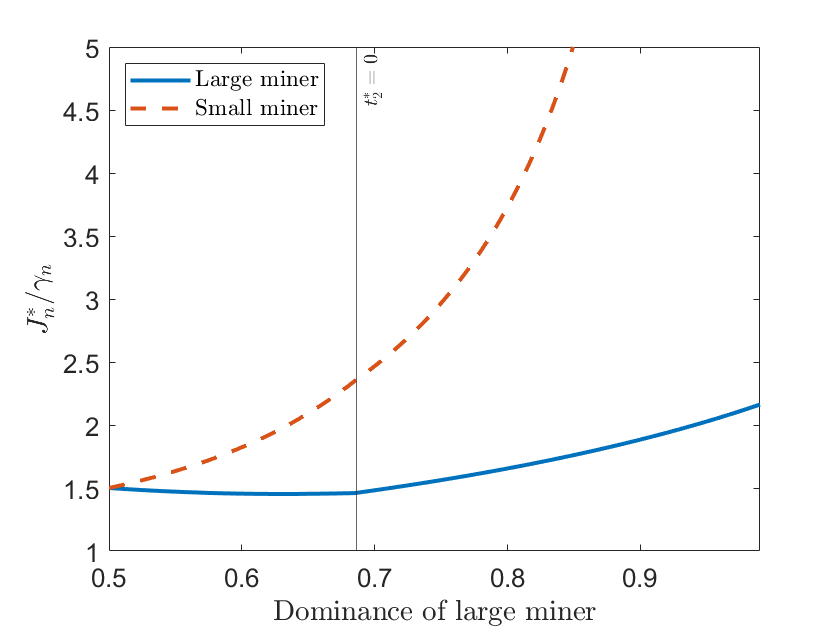}
  \captionof{figure}{$J_n^*/\gamma_n$ at steady-state equilibrium}
  \label{subfig:Dom5}
\end{minipage}%
\begin{minipage}{.5\textwidth}
  \centering
  \includegraphics[width=.9\linewidth]{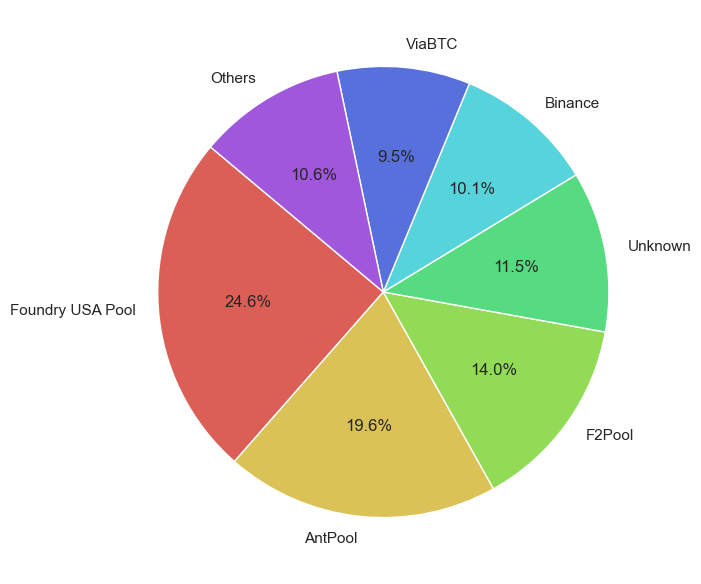}
  \captionof{figure}{Mining pools and their shares}
  \label{fig:2.6.miningpool}
\end{minipage}
\end{figure}

\section{Bitcoin Network}\label{sec:Case Study}

This section applies our model to the Bitcoin network, using data from January 1, 2022, to December 31, 2023. We calibrate the model and explore its implications for the Bitcoin system. The top six mining pools own 89.4\% of the total network capacity, with the largest pool having about 24.6\%. See Figure~\ref{fig:2.6.miningpool}. For simplification, we assume equal capacities for $N=7$ miners, each computed as 70,962,243 terahashes per second (TH/s).


The model incorporates block generation time $F_B$, revenue function $R$, cost rate $c$, and winning rate $w$. Figure~\ref{fig:2.6.ecdf of block} presents the empirical distribution function of 2016 block generation times, with the red line showing the fitted exponential distribution. This visual inspection supports modeling block generation with exponential distributions, assuming no mining gap. Figure~\ref{fig:2.6.hashrate} illustrates fluctuations in computing capacity (TH/s), including sharp declines due to factors like China's mining ban. Close examination reveals no strategic or intentional mining gaps. The estimated total network capacity on December 31, 2023 was $\gamma = 496,735,701$ TH/s.
The total capacity $\gamma$ is estimated using the Bitcoin network's target block generation time of 10 minutes and difficulty information from the distributed ledger. This difficulty determines the probability $p$ of successful mining per hash operation. The network requires $p^{-1}$ hashes or $p^{-1} 10^{-12}$ TH for new block generation, translating to $(\gamma p)^{-1} 10^{-12}$ seconds without mining gaps. Equating this with average sample block generation times yields $\gamma$. With block generation time modeled by an exponential distribution (rate $w \gamma$), we have $w = p ; 10^{12}$.

\begin{figure}
\centering
\begin{minipage}{.48\textwidth}
  \centering
  \includegraphics[width=1\linewidth]{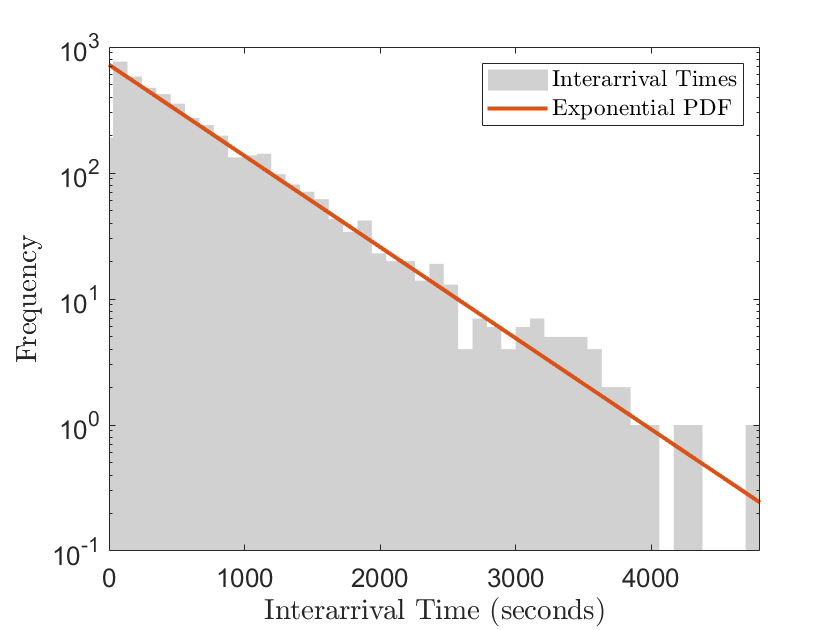}
  \captionof{figure}{Empirical distribution function of block generation time in a log scale.}
  \label{fig:2.6.ecdf of block}
\end{minipage}%
\hspace{5mm}
\begin{minipage}{.48\textwidth}
  \centering
  \includegraphics[width=1\linewidth]{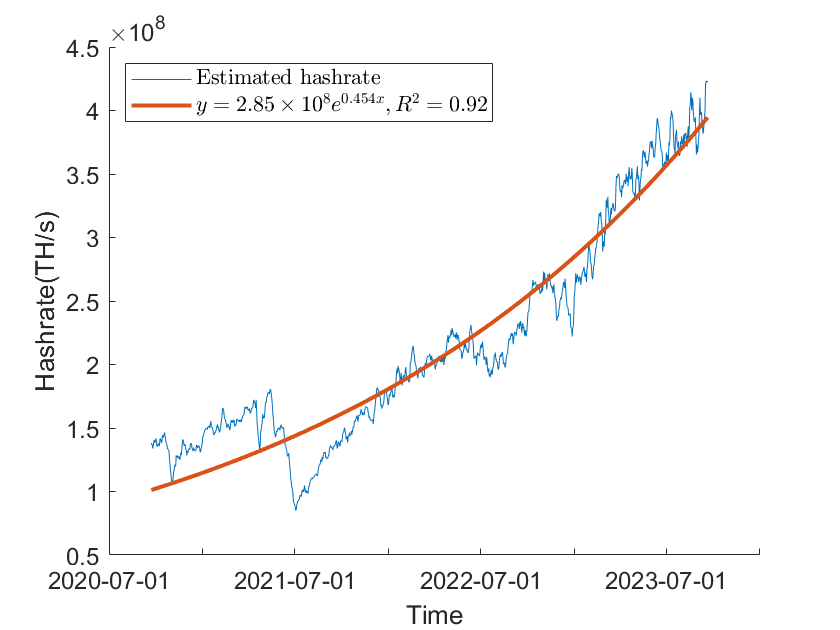}
  \captionof{figure}{Estimated hashrates and the fitted curve.}
  \label{fig:2.6.hashrate}
\end{minipage}
\end{figure}

 

As for the revenue function, actual revenues must be obtained from individual nodes. Since this is hardly feasible, we utilize the real-time mempool data published on the {\tt mempool.observer} site. Transactions arriving at the node operated by the site are shown in Figure~\ref{fig:2.6.Best1MB}. As already noted in the introduction, miners include transactions with the highest fees per unit size (in terms of {\tt vByte}) in a block with a 1MB limit. The vertical axis of the figure shows those best fees (normalized by transaction sizes). Periodic sharp declines happen at moments of block generation. We see that Assumption~\ref{assumption:revenue} reasonably holds. For the sake of tractability, we assume a linear revenue $R(t) = \alpha t + \beta$.
To estimate the parameter $\alpha$, we face the limitation of not having access to the complete historical queue. Thus, we approximate $\alpha$ by using daily transaction fees and the number of new blocks per day. For $\beta$, we use the fixed reward and the market price of a single Bitcoin.
It is important to acknowledge the bias introduced by the presence of up to 1MB of transaction fees in the mempool immediately after block generation. Although these fees contribute to $\beta$, they are excluded from its estimation due to their negligible proportion relative to the block reward. Conversely, while not part of $\alpha$, their inclusion in its calculation leads to a slight overestimation. 

    \begin{figure}[htb!]
        \centering
        \includegraphics[width=\textwidth]{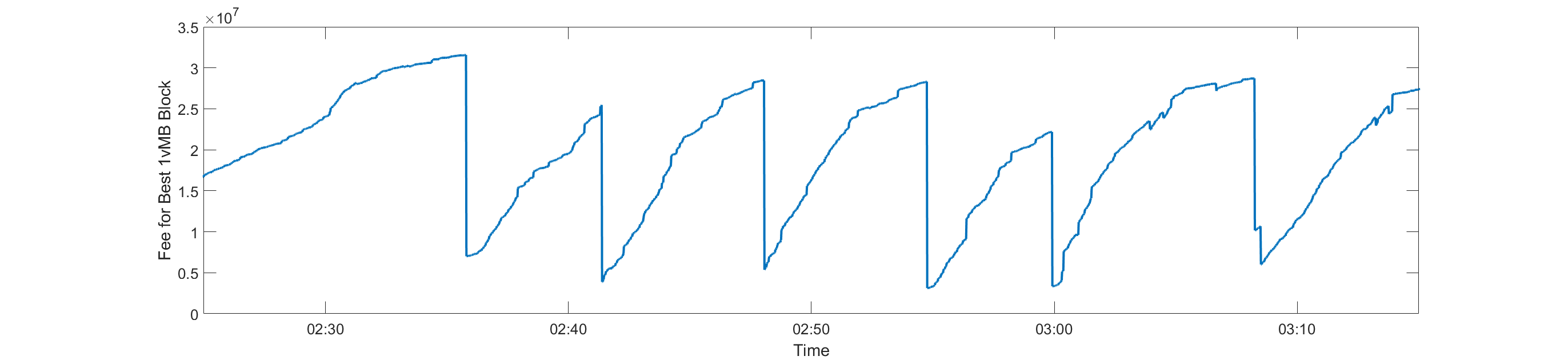}
        \caption{Sum of the fees included in the best 1MB.}
        \label{fig:2.6.Best1MB}
    \end{figure}
    
The last component is the operating cost parameter $c$. There are bottom-up and top-down methods for estimating this. The former means we estimate the cost from electricity rates and the performance of mining machines, but this approach has the disadvantage that there is too much variation among machines, countries, and companies. The other top-down approach, instead, uses break-even points reported by miners. 
Based on \cite{hayes2019}, average production costs of Bitcoin have been reported, e.g., by \cite{jpmorgan}. JP Morgan's estimates, revised in September 2023, come from the data of the Cambridge Bitcoin Electricity Consumption Index, showing that it costs around \$18,000 to mine one Bitcoin. Given the amount of BTC per block, the operating cost is set equal to $cw^{-1}=\$112,500$.
    
Our previous analysis in the case of homogeneous miners with an affine revenue function is applicable to this blockchain system with those input values. In particular, we find that there is no mining gap if the economic reward right after block generation is greater than or equal to $\$112,500 + \alpha\mu/6$. If $\beta$ does not exceed this threshold, then we have two cases. If $\beta$ is smaller than $\alpha\mu$, then there is a steady state winning rate $w$ which is $\varepsilon$-stable for suitable $\varepsilon$. But if $\beta$ is greater than $\alpha\mu$ (but smaller than the threshold for mining gap), then the winning rates by the DAA may not converge and the system does not achieve the target block generation time. Figure~\ref{fig:2.6.Real Data alpha beta plane} shows the division of $\alpha\mu$-$\beta$ plane into the three regions of `no mining gap', `mining gap with no system stability', and `stable mining gap', and the corresponding data points for three years. Actual $\alpha\mu$ values are computed as daily transaction fees divided by the number of blocks for each date, and $\beta$ is 6.25 BTC multiplied by the average daily price of Bitcoin. Note that the dotted lines are computed using the data from December 31, 2023, and the production cost uses the estimate from JP Morgan in September 2023. 

   \begin{figure}[htb!]
        \centering
        \includegraphics[width=0.7\textwidth]{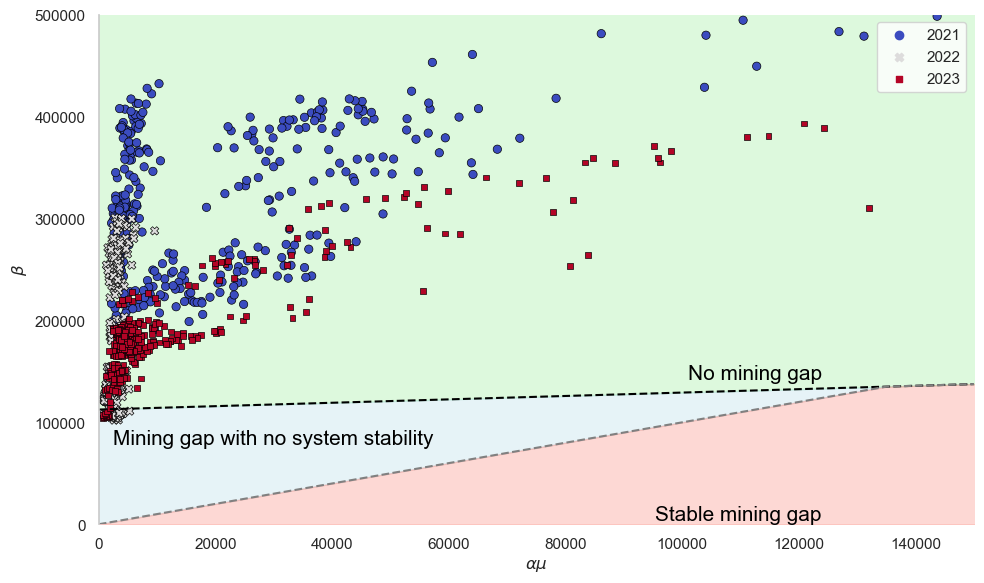}
        \caption{Illustration of the stability of the Bitcoin network. The plane is divided into three regions and the historical data of Bitcoin prices are shown.}
        \label{fig:2.6.Real Data alpha beta plane}
    \end{figure}

According to our analysis, the figure shows that the Bitcoin network in 2022 and 2023 almost fell within the region of the mining gap with no system stability. The primary motivation for mining, predominantly block rewards rather than transaction fees, is expected to bring down the current point of $(\alpha\mu,\beta)$ particularly at the time of block reward halving. However, we also note that the operating cost in our analysis is converted from dollars to BTC. If the operating cost in fiat money is relatively small compared to Bitcoin or if the Bitcoin price in fiat money increases, then the system may remain in the no mining gap region. 
The graph also suggests that even allowing for mining gaps, maintaining system stability would require transaction fee levels (above \$135,000 per block) that are currently difficult to achieve. Such a goal seems only possible with a highly congested transaction queue, extracting greater revenues from users. 

\section{Concluding Remarks}\label{sec:Conclusion}

In this paper, we analyzed the formation and potential ramifications of mining gaps, a phenomenon that might arise within Bitcoin and analogous cryptocurrencies based on the proof-of-work consensus protocol. Through theoretical analysis together with numerical experiments, we have not only corroborated the findings in the literature about mining gaps but have also derived specific conditions for such gaps. 
Our closed-form expression, particularly in the case of homogeneous miners with an affine revenue function, provided plausible scenarios that result in mining gaps. They are, in brief, contingent upon the sustainability of economic rewards for miners such as block rewards, transaction fees, and operating costs. 

Other findings of our study pointed to the potential instability of a blockchain system when mining gaps are combined with difficulty adjustment algorithms. We showed that the intended stable dynamics of the network might break down in the sense that the sequence of mean block generation times based on updated winning rates (or difficulties) may not converge to the  target block generation time, particularly when transaction fees are low relative to block rewards. 
Our analysis further shed light on the formation of coalitions. We found that coalesced miners have higher profitability and that coalitions may lead to larger mining gaps across the network. 
%
In response to these challenges, a shift in protocol design may be necessary with the objective of attenuating the effects of mining gap. Actually, some fundamental changes in reward mechanisms have been proposed, including the fee-burning mechanism used by Ethereum. Other alternatives are to modify the difficulty adjustment algorithm or to increase transaction throughput per hour.

This study is not without limitations. Most notably, we assume that miners base their decisions on the target cryptocurrency instead of fiat money, earning the two sources of incentives: block rewards and transaction fees. This makes sense in small time windows. However, the price appreciation of the cryptocurrency is another important source of economic rewards. Thus, consideration of such a factor would be a highly valuable future research direction.  Another limitation is that we treat transaction fees as an exogenous component. Actual transactions from users form a priority queueing game and respond to the dynamics of block generations, as discussed in \cite{huberman2021monopoly}. There could be externalities as well that incur more or less engagement of users in the network, such as the price spike of the cryptocurrency or security breaches of the network. Nevertheless, we believe our analysis offers a convenient tool for gauging potential problems and future directions when such an environment changes.

\section*{Acknowledgement}
    This work was supported by the National Research Foundation of Korea (NRF) funded by the Korea Government (MSIT) under Grant RS-2023-00278082.  

\section*{Disclosure of interest}
    There are no relevant financial or non-financial competing interests to report.

\bibliographystyle{plainnat}
\bibliography{Ch2_ref}


\begin{appendix}
\section{Proofs of Lemmas}

\begin{proof}[Proof of Lemma~\ref{lem:existence}]
    Following \cite{dmitruk2005existence}, we verify the total of eight conditions for \eqref{eq:target}, A1 to A8. Among them, A1 to A4, and A6 as well as A7 are trivially satisfied. One condition A5, which is referred to as Filippov condition, translates to that there exists some constant $k$ such that $x s_n \leq k ( x^2 + 1)$ for any $t \geq 0$ and for any $x \in [0, \gamma t]$ and $s_n \in [0, \gamma_n]$. This condition holds if we set $k=\gamma$. 
    
    The last condition A8 requires that an integral $\int_{T_1}^{T_2} (w R(t) - c) s_n(t) e^{- w x(t)} \rd t$ called functional piece converges to zero uniformly in all admissible strategies as $T_1$ and $T_2$ go to infinity. In fact, the limiting condition in Assumption~\ref{assumption:revenue} says $R(t) > c/w$ for sufficiently large $T_1$ and $T_2$. We also have $s(t) \geq \eta > 0$ for all large $t$ values thanks to Assumption~\ref{assumption:liminf}. Let us use the symbol $\tau$ to denote the time beyond which those two inequalities hold. Then, we note 
    \begin{eqnarray*}
    \left| \int_{T_1}^{T_2} (w R(t) - c) s_n(t) e^{- w x(t)} \rd t \right| &\leq&  \int_{T_1}^{T_2} (w R(t)  - c) \gamma_n e^{-w \eta (t- \tau)} \rd t \\
    & \leq & \int_{T_1}^{T_2} ( w R(\tau) + R'(\tau) (t - \tau) - c) \gamma_n e^{-w \eta(t - \tau)} \rd t 
    \end{eqnarray*}
    where we use the fact that $s_n(t)$ is bounded above by $\gamma_n$ and $R(\cdot)$ is concave. The right hand side can be readily shown to be expressed as 
    $$
    A \left[ e^{- w \eta T_1} - e^{-w \eta T_2} \right] + B \left[ T_1 e^{-w \eta T_1} - T_2 e^{-w \eta T_2} \right].
    $$
    Since this value converges to zero as $T_1, T_2$ increase, A8 holds. The existence of an optimal control follows from the main theorem of \cite{dmitruk2005existence}. 
    \qed
\end{proof} 
\begin{proof}[Proof of Lemma~\ref{lem:long_term_s}]
    If $w R(0) < c$, there exists a unique $\tau$ such that $wR(\tau) = c$. For any admissible control $s_n$, define $\tilde s_n(t) = s_n(t) {\bf 1}{{ t \geq \tau }}$ and their corresponding utilities are denoted by $J_n$ and $\tilde J_n$. We show $J_n \leq \tilde J_n$:
    \begin{eqnarray*}
    J_n - \tilde J_n &=& \int_0^\infty (w R(t) -c) s_n(t) e^{-w x(t)} \rd t  - \int_0^\infty (w R(t) - c) \tilde s_n(t) e^{-w \tilde x(t)} \rd t \\
    &=& \int_0^\tau (w R(t) -c) s_n(t) e^{-w x(t)} \rd t + \int_\tau^\infty (w R(t) - c) s_n(t) \lt( e^{ - w x(t)} - e^{-w \tilde x(t)} \rt) \rd t. 
    \end{eqnarray*}
    Both terms are negative, implying $J_n \leq \tilde J_n$. Thus, any optimal control $s_n^*$ remains optimal when $s_n^*(t) = 0$ for $t < \tau$.
    
    The first term is negative because $w R(t) < c$ on $[0, \tau)$ and the second term is also negative because $w R(t) > c$ on $(\tau, \infty)$ but $x(t) = \int_0^t s_n(u) \rd u \geq  \int_\tau^t s_n(u) \rd u = \int_0^t \tilde s_n(u) \rd u = \tilde x(t)$.  Both terms are negative, implying $J_n \leq \tilde J_n$. Thus, any optimal control $s_n^*$ remains optimal when $s_n^*(t) = 0$ for $t < \tau$. In fact, if an admissible control $s_n$ is positive on a subset of $[0, \tau)$ with a positive measure, then we immediately see that the control is suboptimal. 
    Note that if $w R(0) \geq c$, then there is nothing to prove as $\{t | w R(t) < c \}$ is empty thanks to Assumption~\ref{assumption:revenue}. 

    For the limiting value of an optimal control, we derive:
    \begin{eqnarray*}
    \sigma_n'(t) &=&  w R'(t) e^{-w x^*(t)} -w (w R(t) - c) s^*(t) e^{-w x^*(t)} + \psi'(t) \\
    &=& \Big[ w R'(t) - w s_{-n}(t) (w R(t) - c) \Big] e^{-w x^*(t)} \\
    &=& w \xi_n(t) e^{-w x^*(t)}.
    \end{eqnarray*}
    Integrating this over $[t, T]$, we obtain 
    $\sigma_n(t)  = \sigma_n(T) - w \int_t^T \xi_n(u) e^{-w x^*(u)} \rd u.$
    By sending $T$ to infinity, we get $\sigma_n(t) = - w \int_t^\infty \xi_n(u) e^{-w x^*(u)} \rd u$. 

    Using Assumption~\ref{assumption:revenue} and the properties of $R(t)$, we show $\xi_n(u)$ becomes negative for large $u$. If $\lim_u R'(u) > 0$, then $\liminf_u (w R(u) - c) s_{-n}(u) = \infty$. If $\lim_u R'(u) = 0$, then $\liminf_u (w R(u) - c) s_{-n}(u) > 0$. This implies $\sigma_n(t)$ is positive for sufficiently large $t$, thus $s_n^*(t) = \gamma_n$. \qed
\end{proof}
\begin{proof}[Proof of Lemma~\ref{lem:singular_interval}]
    Consider a singular control occurring in an interval $(t_1,t_2)$, where $\sigma_n(t)=0$. Define $\displaystyle J_{I}=\int_I (wR(t)-c)s_n(t)e^{-wx(t) } \rd t$ for interval $I$, and let $J_n^*$ and $J_I^*$ denote the optimal utility of miner $n$ over the entire interval and over $I$, respectively. We have $J_n = J_{[0,t_1]} + J_{(t_1,t_2)} + J_{[t_2, \infty)}$. Changes in $s_n$ on $(t_1,t_2)$ affect only $J_{(t_1,t_2)}$ and $J_{[t_2, \infty)}$ via $x_n(t_2)$. 

    Figure~\ref{fig: Singular control} illustrates two possible trajectories $\Gamma_1$ and $\Gamma_2$ for $x_n(t)$ on $(t_1, t_2)$. Define $\Gamma = \Gamma_1 - \Gamma_2$ and let $D$ be the region bounded by $\Gamma$. Abusing notation, the utility $J_n$ for the curve $\Gamma_i$ is denoted by $J_{\Gamma_i}$. Then, $J_{\Gamma_1} - J_{\Gamma_2}$ becomes the line integral over $\Gamma$ where the path of integration is counterclockwise. Applying Green's Theorem:
    \begin{eqnarray*}
    J_{\Gamma_1} - J_{\Gamma_2} &=& \oint_\Gamma (w R(t) -c) s_n(t) e^{- w x(t)} \rd t \\
    &=& \oint_\Gamma (w R(t)  - c) e^{ -w (x_{-n}(t) + x_n )} \rd x_n \\
    &=& \iint_D \frac{\partial}{\partial t} \left[ (w R(t)  - c) e^{ -w (x_{-n}(t) + x_n )} \right] \rd t \rd x_n \\
    &=& \iint_D e^{-w x(t)} w \xi_n(t) \rd t \rd x_n. 
    \end{eqnarray*}
    Since $\sigma_n(t) = 0$ implies $\sigma_n'(t) = 0$, and $\sigma_n'(t) = w \xi_n(t) e^{- w x^*(t)}$, we have $\xi_n(t) = 0$ on $(t_1, t_2)$. Thus, $J_{\Gamma_1} = J_{\Gamma_2}$, indicating that any control $x_n(t)$ yields the same utility given fixed endpoints at $t_1$ and $t_2$.
    
    \begin{figure}[t]
    \includegraphics[width=0.5\textwidth]{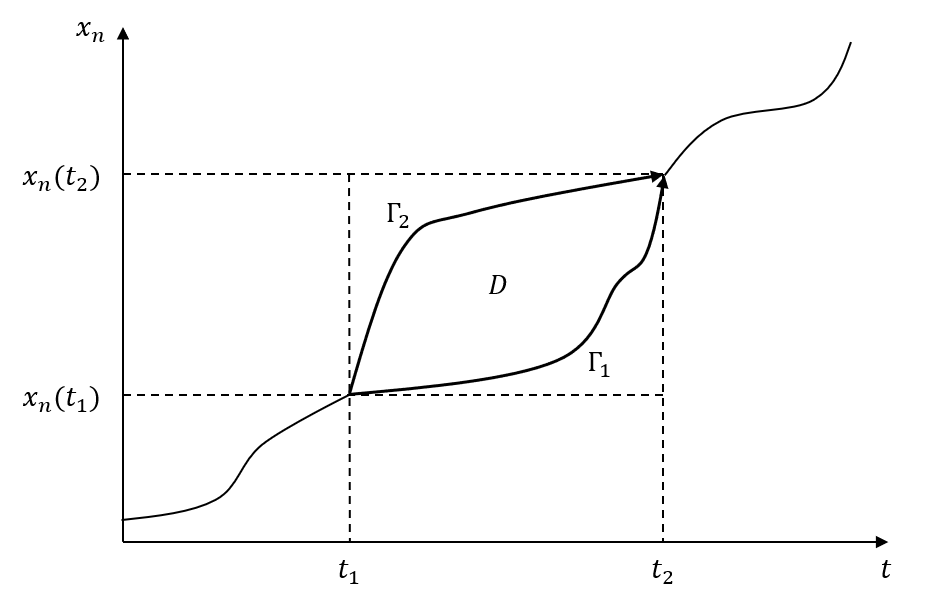}
    \centering
    \caption{Two possible trajectories for $x_n(t)$ and region $D$ by simple closed curve $\Gamma = \Gamma_1 - \Gamma_2$. }
    \label{fig: Singular control}
    \end{figure}

    On the other hand, we have $\psi(t_2) = - w J_{[t_2, \infty)}^*$ according to the definition of $\psi$. Since we have $\sigma_n(t_2) = 0$ and $\sigma_n(t) = (w R(t) - c) e^{-w x^*(t)} + \psi(t)$, we have 
    $$
    J_{[t_2, \infty)}^* =\frac{1}{w} \left( w R(t_2) - c \right) e^{-w x^*(t_2)}.
    $$
    Now, consider a control $s_n$ on $(t_1, t_2)$ defined as $s_n(t) = \Delta \varepsilon^{-1} {\bf 1}_{ \{t_2 - \varepsilon \leq t \leq t_2 \}}$ for some $\Delta$, with $s_n(t) = s_n^*(t)$ outside $(t_1, t_2)$. We can show: 
    \begin{eqnarray*}
    J_{[t_2,\infty)} &=& \int_{t_2}^\infty (w R(t) - c) s_n(t) e^{- w x(t)} \rd t \\
    &=& e^{-w x_n(t_2)} \int_{t_2}^\infty (w R(t) - c) s_n^*(t) e^{- w x_{-n}(t)  - w \int_{t_2}^t s_n^*(u) \rd u} \rd t  \\
    &=& e^{ w (x_n^*(t_2) - x_n(t_2) ) } J_{[t_2, \infty)}^* \\
    &=& \frac{1}{w} (w R(t_2) - c) e^{ - w x(t_2)  }. 
    \end{eqnarray*}  
    For $J_{(t_1, t_2)}$, we derive:
    \begin{eqnarray*}
    J_{(t_1, t_2)} &=& \int_{t_1}^{t_2} (w R(t) - c) s_n(t) e^{-w x(t)} \rd t \\
    &=& \frac{\Delta}{\varepsilon} \int_{t_2 - \varepsilon}^{t_2} (w R(t) - c)  e^{-w x_{-n}(t)} e^{ - w x_n^*(t_1)} e^{ - w \Delta \varepsilon^{-1}(t - t_2+ \varepsilon) } \rd t. 
    \end{eqnarray*}   
    This quantity is upper bounded by : 
    \begin{eqnarray*}
    \lefteqn{ \frac{\Delta }{\varepsilon } (w R(t_2) - c) e^{-w x_{-n}(t_2-\varepsilon)} e^{-w x_n^*(t_1)} \int_{t_2 - \varepsilon}^{t_2}  e^{ - w \Delta \varepsilon^{-1}(t - t_2+ \varepsilon) } \rd t } && \\
    & = & \frac{\Delta }{\varepsilon } (w R(t_2) - c) e^{-w x_{-n}(t_2-\varepsilon)} e^{-w x_n^*(t_1)} \frac{\varepsilon}{w \Delta} \left[ 1 - e^{- w \Delta} \right] \\
    & = &  \frac{1}{w} (w R(t_2) - c) e^{-w x_{-n}(t_2-\varepsilon)} e^{-w x_n^*(t_1)}  \left[ 1 - e^{- w \Delta} \right].
    \end{eqnarray*}
    
    Likewise, it is lower bounded by $w^{-1} (w R(t_2-\varepsilon) - c) e^{-w x_{-n}(t_2)} e^{-w x_n^*(t_1)}  \left[ 1 - e^{- w \Delta} \right]$. Since $J_{(t_1, t_2)}$ is not affected by the choice of $\varepsilon$ and $R(t)$, $x_{-n}(t)$ are continuous, we have
    $$
    J_{(t_1, t_2)} = \frac{1}{w} (w R(t_2) - c) e^{- w x_{-n}(t_2)} e^{-w x_n^*(t_1)} \left[1 - e^{- w \Delta}\right].
    $$
    Finally, by noting that $x(t_2) = x_{-n}(t_2) + x_n(t_2) = x_{-n}(t_2) + x_n^*(t_1) + \Delta$, it is straightforward to check that $J_{(t_1,t_2)} + J_{[t_2, \infty)}$ does not depend on $x_n(t_2)$. Therefore, we conclude that any choice of $s_n(t)$ on $(t_1, t_2)$ is optimal. 
    \qed
\end{proof} 

\begin{proof}[Proof of Lemma~\ref{lem:increasing_s}]
    First, consider the case where $\{t | w R(t) < c \}$ is nonempty. As in Lemma~\ref{lem:long_term_s}, there exists a unique $\tau$ such that $w R(\tau) = c$, and we set $s_n^*(t) = 0$ on $[0, \tau)$. The adjoint variable $\psi$ satisfies: 
    $$
    \psi(t) = - \int_\tau^\infty w s_n^*(u) (w R(u) - c) e^{-w x^*(u)} \rd u. 
    $$
    Thus, $\psi(t)$ is constant on $[0, \tau)$ and equal to $\psi(0)$, which is strictly negative. Consequently, $\sigma_n(t) = (w R(t) - c)e^{- w x^(t)} + \psi(t) < 0$ on $[0, \tau)$.

    Now, $\xi_n(t) = R'(t) - (w R(t) - c) s_{-n}(t)$ is non-increasing, with $\xi_n(0) > 0$ and $\xi_n(t)$ becoming negative for large $t$ as argued in the proof of Lemma~\ref{lem:long_term_s}. It is easy to see that $\xi_n(\cdot)$ cannot be zero at two distinct time points, say $a, b$. If it were, then $R' \equiv 0$ on the interval $(a,b)$, and this is a contradiction to strictly increasing $R$. There must be a single $\tilde\tau > \tau$ such that $\xi_n(t) > 0$ for $t < \tilde\tau$ and $\xi_n(t) < 0$ for $t > \tilde\tau$. Furthermore, $\tilde\tau$ cannot be smaller than or equal to $\tau$ because, otherwise, $\sigma_n(t)$ would start to decrease before $\tau$ and it would stay non-positive, contradicting the conclusion of Lemma~\ref{lem:long_term_s}. 

    As a result, $\sigma_n(t)$ strictly increases from a negative value to a maximum at $\tilde\tau$, then decreases, staying positive for $t \geq \tilde\tau$. We can find a unique $t^*$ such that $\sigma_n(t^*) = 0$, with the switching function changing from negative to positive at this point. The optimal control is thus $s_n^*(t) = \gamma_n {\bf 1}_{\{t \geq t^*\}}$.
    
    For the case where $w R(0) \geq c$, we consider two subcases:
    
    (i) If $\xi_n(0) > 0$, similar arguments apply. There's a unique $\tilde\tau$ where $\xi_n$ changes sign. If $\sigma_n(0) < 0$, there's a unique positive $t^*$ yielding $s_n^*(t) = \gamma_n {\bf 1}_{\{t \geq t^*\}}$. If $\sigma_n(0) \geq 0$, then $t^* = 0$ and $s_n^*(t) = \gamma_n$ for all $t$.
    
    (ii) If $\xi_n(0) \leq 0$, then $\xi_n(t) < 0$ for all $t > 0$, leading to decreasing $\sigma_n(\cdot)$. Since $\sigma_n(t)$ is positive for large $t$, it must be positive throughout, making $t^* = 0$ and $s_n^*(t) = \gamma_n$ for all $t$.
    \qed
\end{proof}
\begin{proof}[Proof of Lemma~\ref{lem:simple_s}]
    From Lemma~\ref{lem:long_term_s}, recall that $\sigma_n(t) = - w \int_t^\infty \xi_n(u) e^{-w x^*(u)} \rd u$. 
    
    When $w R(0) < c$, Lemma~\ref{lem:long_term_s} implies each $t_n^*$ must be at least $\tau = \inf\{t | wR(t) - c = 0\}$. The function $\xi_n(\cdot)$ is decreasing, becoming strictly decreasing for $t \geq \min_i t_i^* \geq \tau$. Since $\xi_n(t) = R'(t) > 0$ on $[0,\min_i t_i^*)$, $\xi_n$ is strictly decreasing on $[\min_i t_i^*, \infty)$, and $\xi_n(t)$ becomes negative for large $t$, there exists a unique $\tilde\tau$ where $\xi_n(\cdot)$ changes sign from positive to negative at $t = \tilde\tau$. As in Lemma~\ref{lem:increasing_s}, $\sigma_n(\cdot)$ is negative on $[0, \tau)$. The behavior of $\xi_n(\cdot)$ implies $\sigma_n'(t) > 0$ for $t < \tilde\tau$ and $\sigma_n'(t) < 0$ for $t > \tilde\tau$. Thus, $\sigma_n(\cdot)$ is strictly increasing on $[0, \tilde\tau)$, maximal at $\tilde\tau$, and then strictly decreasing. This ensures a unique $t_n^*$ such that $\sigma_n(t_n^*) = 0$, as $s_n^* \equiv 0$ would be contradictory.
    
    For $w R(0) \geq c$, we follow Lemma~\ref{lem:increasing_s}. If $\xi_n(0) > 0$, there's a unique positive $t_n^*$ solving $\sigma_n(t_n^*) = 0$ when $\sigma_n(0) < 0$; otherwise, $t_n^* = 0$. If $\xi_n(0) \leq 0$, $\sigma_n(\cdot) > 0$ always, so $t_n^* = 0$.
    
    Now, consider for $m < n$:
    \begin{eqnarray*}
    \lefteqn{ \sigma_n(t_n^*) - \sigma_m( t_n^*)} && \\
    &=& - w \int_{t_n^*}^\infty \big( \xi_n(u) - \xi_m(u) \big) e^{-w x^*(u)} \rd u \\
    &=& w \int_{t_n^*}^\infty (w R(u) - c) \left( - \gamma_n + \gamma_m {\bf 1}_{\{ t \geq t_m^*\}} \right) e^{-w  x^*(u)} \rd u. 
    \end{eqnarray*}
    
    If $t_m^* > t_n^*$, this equals:
    $$
    - w \int_{t_n^*}^{t_m^*} (wR(u) - c) \gamma_n e^{-w x^*(u)} \rd u - w \int_{t_m^*}^\infty (w R(u)-c)(\gamma_n - \gamma_m) e^{-w x^*(u)} \rd u.  
    $$
    This is strictly negative, implying $\sigma_m(t_n^*) > \sigma_n(t_n^*) = 0$. However, $\sigma_m(t_m^*) = 0$ and $\sigma_m(\cdot)$ is negative for $t < t_m^*$, a contradiction. Thus, $t_m^* \leq t_n^*$. This holds even when there's no solution to $\sigma_n(t) =0$, leading to $t_m^*= t_n^* = 0$. 
    \qed

\end{proof}

\begin{proof}[Proof of Lemma~\ref{lem:nash compare}]
As noted in Lemma~\ref{lem:simple_s}, $\sigma_n(t) = -w \int_t^\infty \xi_n(u) e^{-w x^*(u)}\rd u$. Here, $\xi_n(u) = R'(u) - (wR(t)-c)s_{-n}^*(u)$. From the same lemma, we know that $\xi_n(\cdot)$ is positive on $[0, \tilde \tau)$ and negative on $(\tilde\tau, \infty)$. From the relation $\sigma_n'(t) = w\xi_n(t) e^{-wx^*(t)}$, we concluded that $\sigma_n(t)$ increases from a negative value to a maximum positive value at $\tilde\tau$, and stays positive (at least for the Nash equilibrium in Theorem~\ref{thm:nash}). In this case, therefore, the optimal starting time $t_n^*$ is the time $t$ before $\tilde\tau$ such that $\sigma_n(t)$ crosses zero. For miner $n$ with the same capacity $\gamma_n$ but in a different network with greater $\tilde s_{-n}^*(u)$, the function $\tilde \xi_n(u) = R'(u) - (wR(u)-c) \tilde s_{-n}^*(u)$ is not greater than $\xi_n(u)$ at all times. This implies a smaller $\tilde \xi_n$ hits zero at a smaller time value than $\xi_n$, which in turn leads to a smaller $\tilde t_n^*$. \qed

\end{proof}

\section{Proofs of Theorems}

\begin{proof}[Proof of Theorem~\ref{thm:nash}]
    Lemma~\ref{lem:long_term_s} ensures the existence of $\bar t$ such that $s_n^*(t) = \gamma_n$ for all $t \geq \bar t$ and all $n$. Define $\Gamma_n$ as the largest interval including $(\bar t, \infty)$ where $s_n^*(t) = \gamma_n$, and let $\bar t_n = \inf \Gamma_n$. As there are finitely many miners, $\min \bar t_n$ is well-defined.
    
    Suppose piecewise controls $s_n^*$'s form a Nash equilibrium. By Lemma~\ref{lem:constant_s} below, each $s_n^*$ is 0 or $\gamma_n$. We argue by contradiction that $s_n^*$ cannot jump from $\gamma_n$ to 0. Let $t_{\rm off}$ be the last time any miner turns off machines (if such time exists). Assume miner $n$ turns off at $t_{\rm off}$. Then $s_{-n}(t)$ increases for $t \geq t_{\rm off}$, making $\xi_n(t)$ decreasing.
    
    We claim $\sigma_n'(t_{\rm off}) \leq 0$. If $\sigma_n'(t_{\rm off}) > 0$, $\sigma_n$ is strictly increasing near $t_{\rm off}$. This leads to a contradiction with the assumption that $s_n^*$ changes from $\gamma_n$ to 0 at $t_{\rm off}$. Thus, $\xi_n(t_{\rm off}) \leq 0$ since $\sigma_n'(t) = w \xi_n(t)e^{-w x^*(t)}$. Similarly, $\xi_n(\bar t_n) \geq 0$.
    
    Next, we claim $\sigma_n'(t_{\rm off}) \leq 0$. Suppose $\sigma_n'(t_{\rm off}) > 0$. Then, in a neighborhood of $t_{\rm off}$, $\sigma_n$ is strictly increasing. Also, there exist $t_a < t_{\rm off} < t_b$ such that $\sigma_n(t_a) \leq 0$ and $\sigma_n(t_b) > 0$ or $\sigma_n(t_a) < 0$ and $\sigma_n(t_b) \geq 0$. (Due to the property of a switching function.)  The former case implies $s_n^*(t_b) = \gamma_n$ and the latter implies $s_n^*(t_a) = 0$. This is a contradiction since $s_n^*$ is assumed to change from $\gamma_n$ to 0 at $t_{\rm off}$. 
    
    Consequently, $\xi_n(t) = 0$ on $[t_{\rm off}, \bar t_n]$, implying $s_{-n}(t) = 0$ on that interval. This leads to $R'(t) = 0$, contradicting the assumption that $R$ is strictly increasing. We conclude that each $s_n^*$ is non-decreasing. Lemma~\ref{lem:increasing_s} then implies $s_n^*(t) = \gamma_n {\bf 1}_{\{ t \geq t^*_n\}}$ for some $t_n^*$.
    
    Conversely, consider controls $s_n^*$ defined as these simple step functions. As this makes each $s_{-n}$ non-decreasing, Lemma~\ref{lem:increasing_s} suggests this as a Nash equilibrium candidate. The problem reduces to optimizing $t_n^*$'s, which by Lemma~\ref{lem:simple_s}, have unique solutions satisfying $\sigma_n(t_n^*) = 0$.
    \qed
\end{proof}

\begin{lem}\label{lem:constant_s}
    Suppose that $s_{-n}$ is piecewise constant. Then, there exists an optimal control $s_n^*$ that takes values 0 or $\gamma_n$ only.
\end{lem}
\begin{proof}
    Suppose that a singular control occurs on a set of positive measure on some compact interval. Given that $s_{-n}$ is piecewise constant by assumption, it can take only finitely many constant values. Thus, there must exist an interval where $\sigma_n(t) = 0$ and $s_{-n}(t) = k$ (constant) for infinitely many $t$ values.

    Applying Roll's Theorem, we conclude that $\sigma_n'(t) = 0$ must occur infinitely often in this interval. Recall that $\sigma_n'(t) = w \xi_n(t) e^{-w x^*(t)}$, implying the existence of a sequence of $t$ values where $\xi_n(t) = 0$, or equivalently, $R'(t) - (w R(t) -c ) k = 0$.

    This leads to a contradiction unless $k=0$, as $R'(t) - (w R(t) -c )k$ is strictly decreasing and cannot be zero infinitely often for $k \neq 0$. However, if $k = 0$, then $R'(t) = 0$ infinitely many times, implying $R'$ must be zero on an interval (due to its non-increasing and non-negative properties). This contradicts the assumption that $R$ is strictly increasing.
    \qed 
\end{proof}

\begin{proof}[Proof of Theorem~\ref{thm:gap_condition}]
    Without loss of generality, assume miner $N$ is the largest. First, $w R(0) \geq c$ is necessary; otherwise, by Lemma~\ref{lem:long_term_s}, $s_N^*(t) = 0$ for $t < \tau = \inf{t | wR(t) = c}$, implying $t_N^* > 0$. At equilibrium, $$ \xi_N(0) = R'(0) - (w R(0) - c) s_{-N}^*(0) $$ and $\sigma_N(t) = - w \int_t^\infty \xi_N(u) e^{- w x^*(u)} \rd u$ from Lemma~\ref{lem:long_term_s}.
    
    For $t_N^* = 0$, $\sigma_N(0) \geq 0$ is necessary; otherwise, $s_N^*(\cdot) = 0$ near zero, contradicting $t_N^* = 0$. For $t \geq t_N^*$, 
    \begin{eqnarray*}
    \sigma_N(t) &=& w   \int_{t}^\infty \left((wR(u)-c)\gamma_{-N}-R'(u)\right)e^{-wx^*(u)}\rd u \\
    &=& \frac{\gamma_{-N}}{\gamma}\int_{t}^\infty \left(w R(u)-c \right) w \gamma e^{-wx^*(u)} \rd u
     - w \int_{t}^\infty R'(u)e^{-wx^*(u)} \rd u \\
     &=& k \frac{\gamma_{-N}}{\gamma} \int_t^\infty \left(w R(u) - c \right) f_{\bar B}(u) \rd u  - k w  \int_t^\infty R'(u) \bar F_{\bar B}(u) \rd u \\
    &=& k \frac{\gamma_{-N}}{\gamma}  \mathbb{E}\left[ \left(w R(\bar B) - c \right){\bf 1}_{\bar B \geq t}\right]
    + k w R(t) \bar F_{\bar B}(t)   - k w \mathbb{E}\left[ R(\bar B){\bf 1}_{\bar B \geq t} \right]
    \end{eqnarray*}
    where $\gamma_{-N} = \sum_{i=1}^{N-1} \gamma_i$, $\gamma = \sum_{i=1}^N \gamma_i$, and $k = e^{-w \int_0^{t_N^*} s^*(u) \rd u } e^{w \gamma t_N^*}$.
    If $t_N^* = 0$, $\sigma_N(0) \geq 0$ is equivalent to
    $$
    \frac{\gamma_{-N}}{\gamma} \mathbb{E}\left[ (w R(\bar B) - c) \right] + w R(0) - w \mathbb{E}\left[ R(\bar B)\right] \geq 0.
    $$
    which leads to \eqref{eq:gap_condition}.

    For sufficiency, suppose $t_N^* > 0$. Then $\sigma_N(t_N^*) = 0$ (Lemma~\ref{lem:simple_s}). Using the conditional distribution of $\bar B$,
    $$
    \frac{\gamma_{-N}}{\gamma} \mathbb{E}\left[ w R(\bar B + t_N^*) - c \right] + w R(t_N^*) - w \mathbb{E}\left[ R(\bar B + t_N^*)\right] = 0.
    $$
    This is equivalent to ${\sf F}(t_N^*) = 0$ where
    $$
    {\sf F}(t) = \gamma (w R(t) - c) - \gamma_N \mathbb{E}\left[ w R(\bar B + t) - c\right].
    $$
    Note ${\sf F}'(t) = \gamma w R'(t) - \gamma_N \mathbb{E}\left[ w R'(\bar B + t) \right] > 0$ due to Assumption~\ref{assumption:revenue}. Thus, ${\sf F}(0) < 0$, which is equivalent to
    $$
    \frac{\gamma_{-N}}{\gamma} \mathbb{E}\left[ (w R(\bar B) - c) \right] + w R(0) - w \mathbb{E}\left[ R(\bar B)\right] < 0.
    $$
    This completes the proof.
    \qed
\end{proof}

\begin{proof}[Proof of Theorem~\ref{thm:steady_state}]
    Theorem~\ref{thm:nash} shows that optimal responses of miners at equilibrium are simple step functions $\gamma_n {\bf 1}{{t \geq t^*_n}}$. Each $t_n^*$ is a decreasing function of $w$, as higher winning rates accelerate all processes. Let $x^*(t; w)$ denote $x^*$. Then, for $w{\rm new} \geq w_{\rm old}$, we have $x^*(t; w_{\rm new}) \geq x^*(t; w_{\rm old})$, implying $\bmu(w_{\rm new}) \leq \bmu(w_{\rm old})$ from \eqref{eq:bmu}.
            
     We claim that each $t_n^*$ is continuously decreasing in $w$. Assume $\gamma_1 \leq \gamma_2 \leq \cdots \leq \gamma_N$ without loss of generality. By Lemma~\ref{lem:simple_s}, $t_n^*$ is the unique solution to $\sigma_n(t) = 0$ and $t_1^* \leq \cdots \leq t_N^*$. The largest $t_N^*$ is determined by:
     $$
     \int_{t}^\infty \xi_N(u) e^{-w x^*(u)} \rd u = 0
     $$
     where $\xi_N(t) = R'(t) - (w R(t) - c) s_{-N}(t) = R'(t) - (w R(t) -c) \gamma_{-N}$ for $t \geq t_N^*$ and 
     $$
     x^*(t) = \int_0^t s^*(u) \rd u  = \sum_{n=1}^N \gamma_n ( t - t_n)^+. 
     $$
    This equation is equivalent to $ {\sf G}_N(t, w) := \int_{t}^\infty \xi_N(u) e^{- w \gamma u } \rd u = 0$. Note that $\partial_t {\sf G}_N = - \xi_N(t_N^*) e^{- w \gamma t_N^*}$ cannot be zero, ensuring $t_N^*$ is continuous and differentiable in $w$ by the Implicit Function Theorem.
    
    The next largest $t_n^*$ ($t_{N-1}^*$ if $\gamma_{N-1} < \gamma_N$) can be computed using $t_N^*$. For convenience, assume that it is indeed $t_{N-1}^*$. We can proceed similarly by defining a function ${\sf G}_{N-1}(t, w) := \int_t^{t_N^*} \xi_{N-1}(u) e^{- w \gamma_{-N} u} \rd u + e^{w \gamma_N t_N^*} \int_{t_N^*}^\infty \xi_{N-1}(u) e^{-w \gamma u} \rd u$. The desired $t_{N-1}^*$ is a solution to ${\sf G}_{N-1}(t,w) = 0$. We can check $\partial_t {\sf G}$ cannot be zero at $t = t_{N-1}^*$ for the same reason as above. Hence, $t_{N-1}^*$ is also continuously decreasing and differentiable in $w$. This argument can be repeated for all other $t_n^*$'s. One consequence of this continuity and \eqref{eq:bmu} is that $\bmu(w)$ is  a differentiable function of $w$. 
            
    If there's no mining gap when $w = (\gamma \mu)^{-1}$, block distribution $B$ follows an exponential distribution with rate $w\gamma = \mu^{-1}$, satisfying $\bmu(w) = \mu$. If there's a mining gap, the winning rate must be adjusted higher to approach the target $\mu$.
             
    For general $w$, let us define $\bt(w) = \inf\{t : R'(t) \leq (wR(t) - c) \min_n \gamma_{-n} \}$. Here, $\min_n \gamma_{-n}$ is equal to $\gamma_{-N}$ from the assumption on capacity sizes. As in the second paragraph of this proof above, $\xi_N(t_N^*)$ must be positive. The newly defined $\bt(w)$ is therefore greater than $t_N^*$.  This leads to $s^*(t) \geq \gamma {\bf 1}_{\{t \geq \bt(w)\}}$ and thus
    $$
    \bmu(w) = \int_0^\infty \exp\left( - w x^*(t) \right) \rd t \leq \int_0^\infty \exp\left( - w \int_0^t \gamma {\bf 1}_{\{u \geq \bt(w)\}} \rd u \right) \rd t = \bt(w) + \frac{1}{\gamma w}.
    $$
    This inequality shows that $\lim_{w \rightarrow \infty} \bmu(w) \leq \lim_{w \rightarrow \infty}\left( \bt(w) + (\gamma w)^{-1} \right)$. However, the convergence of $\bt(w)$ to 0 is obvious because we have continuous $R(\cdot)$  and positive $R(0)$ (recall that there are always sufficiently large transaction requests in the mempool). Therefore, $\lim_{w \rightarrow \infty} \bmu(w) = 0$.  On the other hand, it is clear that $\lim_{w \rightarrow 0} \bmu(w) = \infty$. By the Intermediate Value Theorem, we conclude that there exists $w$ with  $\bmu(w) = \mu$. 
    \qed
\end{proof}
\begin{proof}[Proof of Theorem~\ref{thm:DAA_conv}]
    Suppose $\bmu'(w) > -2 \mu / w$. We have $f'(x) = \bmu'(w e^x) w e^x / \bmu(w e^x)$ and $f'(0) = \bmu'(w) w / \mu$, so $f'(0) > -2$. As $\lim_{x \rightarrow 0} f(x)/x = f'(0)$, we can find an $\varepsilon$-neighborhood of 0, $B_\varepsilon(0)$, where $f(x_0) / x_0 > -2$ for any $x_0 \in B_\varepsilon(0)$. This implies $x_1 / x_0 > -1$.
    
    On the other hand, we note that if $x_0 > 0$, then $f(x_0) < 0$ (see Figure~\ref{fig:2.4.DAA_conv}) and thus $x_1 /x_0 < 1$. If $x_0 < 0$, then $f(x_0) > 0$ and thus $x_1 / x_0 < 1$ again. Combining this with the observation in the previous paragraph, we get $|x_1/x_0| < 1$. Therefore, $x_1$ stays in $B_\varepsilon(0)$. Reating the same argument, we obtain a decreasing sequence $x_i$'s inside the same neighborhood. Consequently, $\lim_i x_i = 0$. 
    
    Now suppose $\bmu'(w) < -2 \mu / w$ or $f'(0) < -2$. Further suppose that $w$ is $\varepsilon_0$-stable for some $\varepsilon_0$. Let us fix a small positive constant $k$ so that $f'(0) < -2 - k$. Then, there exists some $\varepsilon$ such that $f(x_0) / x_0 < -2  - k$ for any $x_0$ in the $\varepsilon$-neighborhood of the origin while $|w e^{x_0} - w| < \varepsilon_0$. This yields $x_1 / x_0 < -1 - k$. Hence, $|x_1|$ is greater than $|x_0|$. Due to the assumed stability, $|x_i| \leq |x_0|$ for all large $i$'s. This is a contradiction as $x_{i+1}$ is outside of the neighborhood for the same reason. 
    \qed
\end{proof}

\end{appendix}

\end{document}

%% file: symbols.tex
\newcommand{\vs}{\vspace*{3mm}}

\newcommand{\lt}{\left}
\newcommand{\rt}{\right}
\newcommand{\bex}{\begin{eqnarray*}}
\newcommand{\eex}{\end{eqnarray*}}

\newcommand{\rd}{{\rm d}}

\def\bt{{\bf t}}

\def\bw{{\bf w}}

\def\bmu{\mbox{\boldmath $\mu$}}

\def\bbE{\mathbb{E}}

\def\bbP{\mathbb{P}}

%% file: Ch2_ref.bib
@Online{jpmorgan,
 author = {Ossinger, J.},
 year = {2022},
 title = {{JPMorgan} Says Bitcoin Cost of Production May Be Down to \$13,000},
 journal = {Bloomberg},
 url = {https://www.bloomberg.com/news/articles/2022-07-14/jpmorgan-says-bitcoin-cost-of-production-may-be-down-to-13-000}
}

@article{hayes2019,
  title={Bitcoin price and its marginal cost of production: {S}upport for a fundamental value},
  author={Hayes, A. S.},
  journal={Applied Economics Letters},
  volume={26},
  pages = {554--560},
  year={2019}
}

@article{biais2023,
  title={Advances in blockchain and crypto economics},
  author={Biais, B. and Capponi, A. and Cong, L. W. and Gaur, V. and Giesecke, K.},
  journal={Management Science},
  volume={69},
  pages = {6417--6426},
  year={2023}
}

@article{dmitruk2005existence,
  title={Existence theorem in the optimal control problem on an infinite time interval},
  author={Dmitruk, A. V. and Kuz'kina, N. V.},
  journal={Mathematical Notes},
  volume={78},
  pages={466--480},
  year={2005}
}

@inproceedings{eyal2015miner,
  title={The miner's dilemma},
  author={Eyal, I.},
  booktitle={IEEE Symposium on Security and Privacy},
  pages={89--103},
  year={2015}
}

@inproceedings{kwon2019bitcoin,
  title={Bitcoin vs. bitcoin cash: {C}oexistence or downfall of bitcoin cash?},
  author={Kwon, Y. and Kim, H. and Shin, J. and Kim, Y.},
  booktitle={IEEE Symposium on Security and Privacy},
  pages={935--951},
  year={2019}
}

@article{eyal2018majority,
  title={Majority is not enough: {B}itcoin mining is vulnerable},
  author={Eyal, I. and Sirer, E.},
  journal={Communications of the ACM},
  volume={61},
  pages={95--102},
  year={2018}
}

@inproceedings{sapirshtein2017optimal,
  title={Optimal selfish mining strategies in bitcoin},
  author={Sapirshtein, A. and Sompolinsky, Y. and Zohar, A.},
  editor = {Grossklags, J. and Preneel, B.},
  booktitle={Financial Cryptography and Data Security},
  pages={515--532},
  year={2017},
  organization={Springer}
}

@article{hansjoerg2022profitability,
  title={On the profitability of selfish blockchain mining under consideration of ruin},
  author={Albrecher, H. and Goffard, P.-O.},
  journal={Operations Research},
  volume={70},
  pages={179--200},
  year={2022}
}

@inproceedings{carlsten2016instability,
  title={On the instability of bitcoin without the block reward},
  author={Carlsten, M. and Kalodner, H. and Weinberg, S. M. and Narayanan, A.},
  booktitle={Proceedings of the 2016 ACM SIGSAC Conference on Computer and Communications Security},
  pages={154--167},
  year={2016}
}

@inproceedings{tsabary2018gap,
  title={The gap game},
  author={Tsabary, I. and Eyal, I.},
  booktitle={Proceedings of the 2018 ACM SIGSAC Conference on Computer and Communications Security},
  pages={713--728},
  year={2018}
}

@inproceedings{goren2019mind,
  title={Mind the mining},
  author={Goren, G. and Spiegelman, A.},
  booktitle={Proceedings of the 2019 ACM Conference on Economics and Computation},
  pages={475--487},
  year={2019}
}

@inproceedings{fiat2019energy,
  title={Energy equilibria in proof-of-work mining},
  author={Fiat, A. and Karlin, A. and Koutsoupias, E. and Papadimitriou, C.},
  booktitle={Proceedings of the 2019 ACM Conference on Economics and Computation},
  pages={489--502},
  year={2019}
}

@inproceedings{kawase2017transaction,
  title={Transaction-confirmation time for bitcoin: {A} queueing analytical approach to blockchain mechanism},
  author={Kawase, Y. and Kasahara, S.},
  editor={Yue, W. and Li, Q.-L. and  Jin, S. and Ma, Z.},
  booktitle={Queueing Theory and Network Applications},
  pages={75--88},
  year={2017}
}

@article{kasahara2018effect,
  title={Effect of Bitcoin fee on transaction-confirmation process},
  author={Kasahara, S. and Kawahara, J.},
  journal={Journal of Industrial and Management Optimization},
  volume={15},
  pages={365--386},
  year={2018}
}

@article{bowden2020modeling,
  title={Modeling and analysis of block arrival times in the Bitcoin blockchain},
  author={Bowden, R. and Keeler, H. P. and Krzesinski, A. E. and Taylor, P. G.},
  journal={Stochastic Models},
  volume={36},
  pages={602--637},
  year={2020}
}

@inproceedings{gebraselase2021transaction,
  title={Transaction characteristics of bitcoin},
  author={Gebraselase, B. G. and Helvik, B. E. and Jiang, Y.},
  booktitle={IFIP/IEEE International Symposium on Integrated Network Management},
  pages={544--550},
  year={2021}
}

@article{huberman2021monopoly,
  title={Monopoly without a monopolist: {A}n economic analysis of the bitcoin payment system},
  author={Huberman, G. and Leshno, J. D. and Moallemi, C.},
  journal={Review of Economic Studies},
  volume={88},
  pages={3011--3040},
  year={2021}
}

@article{li2022analyzing,
  title={Analyzing Bitcoin transaction fees using a queueing game model},
  author={Li, J. and Yuan, Y. and Wang, F.-Y.},
  journal={Electronic Commerce Research},
  volume = {22},
  pages={135--155},
  year={2022}
}

@article{easley2019mining,
  title={From mining to markets: {T}he evolution of bitcoin transaction fees},
  author={Easley, D. and O'Hara, M. and Basu, S.},
  journal={Journal of Financial Economics},
  volume={134},
  pages={91--109},
  year={2019}
}

@inproceedings{kroll2013economics,
  title={The economics of {B}itcoin mining, or {B}itcoin in the presence of adversaries},
  author={Kroll, J. A. and Davey, I. C. and Felten, E. W.},
  booktitle={The 12th Workshop on the Economics of Information Security},
  number={11},
  year={2013}
}

@inproceedings{brenner2015trends,
author= {M{\"o}ser, M. and B{\"o}hme, R.},
editor={Brenner, M. and Christin, N. and Johnson, B. and Rohloff, K.},
title={Trends, Tips, Tolls: {A} Longitudinal Study of Bitcoin Transaction Fees},
booktitle={Financial Cryptography and Data Security},
year={2015},
pages={19--33}
}

@inproceedings{gervais2016security,
author = {Gervais, A. and Karame, G. O. and W\"{u}st, K. and Glykantzis, V. and Ritzdorf, H. and Capkun, S.},
title = {On the Security and Performance of Proof of Work Blockchains},
booktitle = {Proceedings of the 2016 ACM SIGSAC Conference on Computer and Communications Security},
pages = {3--16},
year = {2016}
}

@inproceedings{nayak2016stubborn,
  author={Nayak, K. and Kumar, S. and Miller, A. and Shi, E.},
  booktitle={IEEE European Symposium on Security and Privacy}, 
  title={Stubborn Mining: {G}eneralizing Selfish Mining and Combining with an Eclipse Attack}, 
  year={2016},
  pages={305--320}
}
